\documentclass[final,leqno,onefignum,onetabnum]{siamltex}

\usepackage{amsmath,amssymb,euscript,yfonts,psfrag,latexsym,dsfont,graphicx,bbm,color,amstext,wasysym,epsfig,parskip,textcomp}
\usepackage{showlabels}
\usepackage{tikz} 
\usepackage{color}
\usepackage{float}

\newtheorem{corrolary}{Corrolary}

\newtheorem{problem}{Problem}

\newcommand{\Id}{{\rm Id}}
\renewcommand{\Lambda}{{\mathcal V}}
\renewcommand{\Gamma}{{\mathcal Y}}

\newcommand{\toll}{{toll}}
\renewcommand{\supp}{{\rm Supp\,}}
\newcommand{\blue}{\color{black}}

\newcommand{\black}{\color{black}}
\definecolor{gray}{rgb}{0.5,0.5,0.5}

\newcommand{\f}{{\rm v}}

\DeclareMathOperator*{\argmin}{arg\,min}

\title{Optimal Transport through a Toll Station}
\author{Arthur Stephanovitch\thanks{\'Ecole normale sup\'erieure Paris-Saclay, France. ({\tt  arthur.stephanovitch@ens-paris-saclay.fr}).} \and Anqi Dong\thanks{University of California, Irvine, CA 92697, USA. ({\tt  anqid2@uci.edu})} \and Tryphon T.\ Georgiou\thanks{University of California, Irvine, CA 92697, USA. ({\tt  tryphon@uci.edu})}}
\date{May 2022}

\def\spacingset#1{\def\baselinestretch{#1}\small\normalsize}
\spacingset{.95}

\begin{document}

\maketitle

\begin{abstract}
    We address the problem of optimal transport with a quadratic cost functional and a constraint on the flux through a constriction along the path. The constriction, conceptually represented by a toll station, limits the flow rate across. We provide a precise formulation which, in addition, is amenable to generalization in higher dimensions. We work out in detail the case of transport in one dimension by proving existence and uniqueness of solution. Under suitable regularity assumptions we give an explicit construction of the transport plan. Generalization of flux constraints to higher dimensions and possible extensions of the theory are discussed.
\end{abstract}

\section{Introduction}
In recent years, the Monge-Kantorovich theory of Optimal Mass Transport has impacted a wide range of mathematical and scientific disciplines from probability theory to geophysics and from thermodynamics to machine learning \cite{gangbo1996geometry,douglas1999applications,ambrosio2005gradient,villani2009optimal,villani2021topics}. Indeed, the Monge-Kantorovich paradigm of transporting one distribution to another, by seeking to minimize a suitable cost functional, has proved enabling in many ways. It gave rise to a class of control problems \cite{chen2021stochastic,chen2021controlling}, underlies variational principles in physics \cite{jordan1998variational,otto2001geometry}, provided natural regularization penalties in inverse problems \cite{bredies2022generalized}, led to new identification techniques in data science \cite{haker2004optimal,peyre2019computational}, in graphical models \cite{fan2022complexity}, and was linked to large deviations in probability theory \cite{leonard2007large,chen2016relation}.

{\blue Historically, the Monge-Kantorovich theory proved especially relevant in economics when physical commodities were the object to be transported--a fact that contributed to L.\ Kantorovich receiving the Noble prize. Extensions that pertain to physical constraints along the transport naturally were soon brought up. For instance, moment-type constraints, have been considered in \cite[Section 4.6.3]{rachev1998mass} and, more recently, far generalized in \cite{ekren2018constrained}. Congestion being a significant impediment to transport has also drawn the attention of theorist and practitioners alike. For instance,
besides optimizing for a transportation, considerations of an added path-dependent cost to alleviate congestion has been considered in
\cite{carlier2008optimal}, see also \cite[Section 4]{santambrogio2015optimal} for a comprehensive study of this research direction. Along a different direction, constraints have been introduced for probability densities as part of the optimization problem. Such bounds can capture the capacity of the transportation medium and as such have been studied 
in \cite{korman2013insights}, or dynamical flow constraints as in \cite{gladbach2022limits}.

In the present work}, we formulate and address a natural variant of the standard optimal mass transport problem by imposing  a hard constraint on the flux rate at a point along the path between distributions. Specifically, we pose and resolve the most basic such problem where the restriction on throughput of the transport plan takes place at a single point. With this constraint in place, we seek to minimize a usual quadratic cost functional.

The analysis we provide focuses on one-dimensional distributions, with transport taking place on $\mathbb R$. We prove existence and uniqueness of an optimal transport plan, and under suitable regularity conditions, give an explicit construction. A slight generalization of our formulation, where the distributions have support on $\mathbb R^d$ but the transportation is to take place through a specified ``constriction'' point, with a similar throughput constraint, can be worked out in the same manner and it is sketched in the concluding remarks. The more general case where the transport takes place on higher dimensional manifolds with the throughput through possibly multiple points, curves, or surfaces  similarly restricted is substantially more challenging and much remains open.

The problem formulation and ideas in the mathematical analysis that follows can be visualized by appealing to Figure \ref{fig:my_label}. We begin with two probability densities $\rho_0, \rho_1$ having support on $\mathbb{R}$ and finite second-order moments, and seek to transport one to the other, $\rho_0$ to $\rho_1$, within a window of time (herein, of duration normalized to $1$) while minimizing a quadratic cost in the local velocity. That is, we seek to minimize the action integral of kinetic energy along the transport path. The minimal cost of the unrestricted transport is the so-called Wasserstein distance $\mathcal W_2(\rho_0,\rho_1)$ (a metric on the space of probability measures); we refer to standard references \cite{villani2009optimal,villani2021topics} for the unconstrained optimal transport problem. 
The schematic in Figure \ref{fig:my_label} exemplifies a constraint at a pre-specified point, $x_0$, that can be seen as the location of constriction, or, of a toll along the transport, where throughput is bounded. That is, the flow rate across $x_0$ for mass times velocity is bounded by a value $h$. A vertical axis pointing downwards at $x_0$ marks the time when a specific mass-element crosses the toll, necessitating at least $1/h$ duration for the unit mass of the probability density $\rho_0$ to go through, in the most favorable case where the throughput rate is maintained for the duration (that is normalized to $1$ time unit). 

In the body of the paper, we prove existence and uniqueness of an optimal transport plan and, assuming suitable regularity of the distributions, we provide an explicit construction for the solution. We further explore consequences of the toll being kept maximally ``busy'' while mass is being transported through, in conjunction with minimizing the quadratic cost criterion on the kinetic energy, and we highlight ensuing properties of the optimal plan.

Specifically, in Section \ref{sec:sec1} we develop the formulation of the flux constraint and give a precise definition of the problem (Problem \ref{prob:problem1}). In Section \ref{sec:sec2} we prove existence  and uniqueness (Theorem \ref{thm:thm1}) of solution, while conveniently recasting the problem in terms of a flux variable (Problem \ref{prob:problem2}).
Section \ref{sec:sec3} deals with the structural form of the transport and properties of solutions. We summarize the basic elements that allow an explicit construction of the solution in Theorem \ref{thm:4.5}. Section \ref{sec:example} provides a rudimentary example of transporting between uniform distributions, that highlights the essential property that speed needs to be suitably adjusted so as to fully utilize the throughput of the toll, while minimizing the quadratic cost. We close (Section \ref{sec:discussion}) with a discussion on possible extension of the problem to higher dimensions and multiple tolls. While the theory may be readily extended in certain cases, much remains to be understood. Such problems are of natural engineering and scientific interest.

\begin{figure}[!htb]
\centering
\tikzset{every picture/.style={line width=0.75pt}} 

\scalebox{0.65}{
\begin{tikzpicture}[x=0.75pt,y=0.75pt,yscale=-1,xscale=1]

\draw [line width=1.5]  (0.97,127.18) -- (574.8,125.58)(288.03,394.77) -- (287.27,123.25) (567.81,130.6) -- (574.8,125.58) -- (567.79,120.6) (283.01,387.78) -- (288.03,394.77) -- (293.01,387.75)  ;
\draw [color={rgb, 255:red, 0; green, 0; blue, 0 }  ,draw opacity=1 ][fill={rgb, 255:red, 74; green, 144; blue, 226 }  ,fill opacity=1 ][line width=1.5]    (59.6,126.9) .. controls (99.6,96.9) and (89.6,48.4) .. (121.6,33.4) .. controls (153.6,18.4) and (130.6,60.9) .. (164.6,67.4) .. controls (198.6,73.9) and (195.6,112.4) .. (212.6,127.4) ;
\draw [color={rgb, 255:red, 0; green, 0; blue, 0 }  ,draw opacity=1 ][fill={rgb, 255:red, 74; green, 144; blue, 226 }  ,fill opacity=1 ][line width=1.5]    (364.1,126.4) .. controls (404.1,96.4) and (395.75,38.6) .. (422.6,57.9) .. controls (449.45,77.2) and (436.6,103.9) .. (466.1,48.9) .. controls (495.6,-6.1) and (482.1,104.73) .. (515.6,125.73) ;
\draw  [color={rgb, 255:red, 74; green, 144; blue, 226 }  ,draw opacity=1 ][fill={rgb, 255:red, 74; green, 144; blue, 226 }  ,fill opacity=1 ][line width=0.75]  (288.87,170.49) -- (334.72,140.15) -- (334.67,303.66) -- (288.81,334) -- cycle ;
\draw  [color={rgb, 255:red, 208; green, 2; blue, 27 }  ,draw opacity=1 ][fill={rgb, 255:red, 208; green, 2; blue, 27 }  ,fill opacity=1 ] (334.76,252.49) -- (334.67,303.95) -- (289.49,337.5) -- (289.58,286.03) -- cycle ;
\draw  [color={rgb, 255:red, 0; green, 0; blue, 0 }  ,draw opacity=1 ][fill={rgb, 255:red, 208; green, 2; blue, 27 }  ,fill opacity=1 ][line width=1.5]  (87.85,89.65) .. controls (92.1,81.4) and (106.6,29.4) .. (104.77,63.32) .. controls (102.93,97.23) and (105.27,128.4) .. (104.77,127.32) .. controls (104.27,126.23) and (63.35,127.15) .. (59.6,126.9) .. controls (55.85,126.65) and (69.6,120.4) .. (77.35,108.4) .. controls (85.1,96.4) and (83.6,97.9) .. (87.85,89.65) -- cycle ;
\draw    (1.95,127.38) -- (572.61,125.38) (58.93,123.18) -- (58.96,131.18)(115.93,122.98) -- (115.96,130.98)(172.93,122.78) -- (172.96,130.78)(229.93,122.58) -- (229.96,130.58)(286.93,122.38) -- (286.96,130.38)(343.93,122.18) -- (343.96,130.18)(400.93,121.98) -- (400.96,129.98)(457.93,121.78) -- (457.96,129.78)(514.93,121.58) -- (514.96,129.58)(571.93,121.38) -- (571.96,129.38) ;
\draw  [dash pattern={on 4.5pt off 4.5pt}]  (334.67,303.95) -- (334.56,385.74) ;
\draw    (289.06,378.34) -- (333,346.03) ;
\draw [shift={(334.62,344.85)}, rotate = 143.68] [color={rgb, 255:red, 0; green, 0; blue, 0 }  ][line width=0.75]    (10.93,-3.29) .. controls (6.95,-1.4) and (3.31,-0.3) .. (0,0) .. controls (3.31,0.3) and (6.95,1.4) .. (10.93,3.29)   ;
\draw    (334.62,344.85) -- (290.67,377.15) ;
\draw [shift={(289.06,378.34)}, rotate = 323.68] [color={rgb, 255:red, 0; green, 0; blue, 0 }  ][line width=0.75]    (10.93,-3.29) .. controls (6.95,-1.4) and (3.31,-0.3) .. (0,0) .. controls (3.31,0.3) and (6.95,1.4) .. (10.93,3.29)   ;
\draw  [dash pattern={on 4.5pt off 4.5pt}]  (334.72,140.15) -- (356.56,140.24) ;
\draw  [dash pattern={on 4.5pt off 4.5pt}]  (334.67,303.95) -- (355.06,303.24) ;
\draw [color={rgb, 255:red, 0; green, 0; blue, 0 }  ,draw opacity=1 ][line width=1.5]    (212.6,127.4) -- (288.87,170.49) ;
\draw [shift={(257.53,152.78)}, rotate = 209.46] [color={rgb, 255:red, 0; green, 0; blue, 0 }  ,draw opacity=1 ][line width=1.5]    (14.21,-4.28) .. controls (9.04,-1.82) and (4.3,-0.39) .. (0,0) .. controls (4.3,0.39) and (9.04,1.82) .. (14.21,4.28)   ;
\draw [color={rgb, 255:red, 0; green, 0; blue, 0 }  ,draw opacity=1 ][line width=1.5]    (104.77,127.32) -- (288.18,287.07) ;
\draw [shift={(202.36,212.32)}, rotate = 221.06] [color={rgb, 255:red, 0; green, 0; blue, 0 }  ,draw opacity=1 ][line width=1.5]    (14.21,-4.28) .. controls (9.04,-1.82) and (4.3,-0.39) .. (0,0) .. controls (4.3,0.39) and (9.04,1.82) .. (14.21,4.28)   ;
\draw [color={rgb, 255:red, 0; green, 0; blue, 0 }  ,draw opacity=1 ][line width=1.5]    (59.6,126.9) -- (288.09,338.53) ;
\draw [shift={(179.57,238.02)}, rotate = 222.81] [color={rgb, 255:red, 0; green, 0; blue, 0 }  ,draw opacity=1 ][line width=1.5]    (14.21,-4.28) .. controls (9.04,-1.82) and (4.3,-0.39) .. (0,0) .. controls (4.3,0.39) and (9.04,1.82) .. (14.21,4.28)   ;
\draw    (355.06,303.24) -- (354.07,142.29) ;
\draw [shift={(354.06,140.29)}, rotate = 89.65] [color={rgb, 255:red, 0; green, 0; blue, 0 }  ][line width=0.75]    (10.93,-3.29) .. controls (6.95,-1.4) and (3.31,-0.3) .. (0,0) .. controls (3.31,0.3) and (6.95,1.4) .. (10.93,3.29)   ;
\draw    (354.06,140.29) -- (355.04,301.24) ;
\draw [shift={(355.06,303.24)}, rotate = 269.65] [color={rgb, 255:red, 0; green, 0; blue, 0 }  ][line width=0.75]    (10.93,-3.29) .. controls (6.95,-1.4) and (3.31,-0.3) .. (0,0) .. controls (3.31,0.3) and (6.95,1.4) .. (10.93,3.29)   ;
\draw  [fill={rgb, 255:red, 208; green, 2; blue, 27 }  ,fill opacity=1 ][line width=1.5]  (402.43,63.43) .. controls (405.93,55.18) and (412.18,53.93) .. (412.93,53.93) .. controls (413.68,53.93) and (413.71,125.33) .. (414.46,125.58) .. controls (415.21,125.83) and (363.43,127.43) .. (365.68,125.43) .. controls (367.93,123.43) and (382.68,111.18) .. (390.68,91.68) .. controls (398.68,72.18) and (398.93,71.68) .. (402.43,63.43) -- cycle ;

\draw (74.44,0.54) node [anchor=north west][inner sep=0.75pt]  [font=\LARGE,rotate=-358.99,xslant=0.22]  {$\rho _{_{0}}$};
\draw (511.68,2.07) node [anchor=north west][inner sep=0.75pt]  [font=\LARGE,xslant=0.09]  {$\rho _{_{1}}$};
\draw (560.79,135) node [anchor=north west][inner sep=0.75pt]  [font=\LARGE]  {$x$};
\draw (257.9,378.84) node [anchor=north west][inner sep=0.75pt]  [font=\LARGE]  {$t$};
\draw (311.84,361.59) node [anchor=north west][inner sep=0.75pt]  [font=\Large]  {$h$};
\draw (362.64,200) node [anchor=north west][inner sep=0.75pt]  [font=\LARGE]  {$\frac{1}{h}$};
\draw (275,100) node [anchor=north west][inner sep=0.75pt]  [font=\LARGE]  {$x_{0}$};
\end{tikzpicture}
}
\caption{Illustration of optimal transport through a toll with finite throughput}
\label{fig:my_label}
\end{figure}
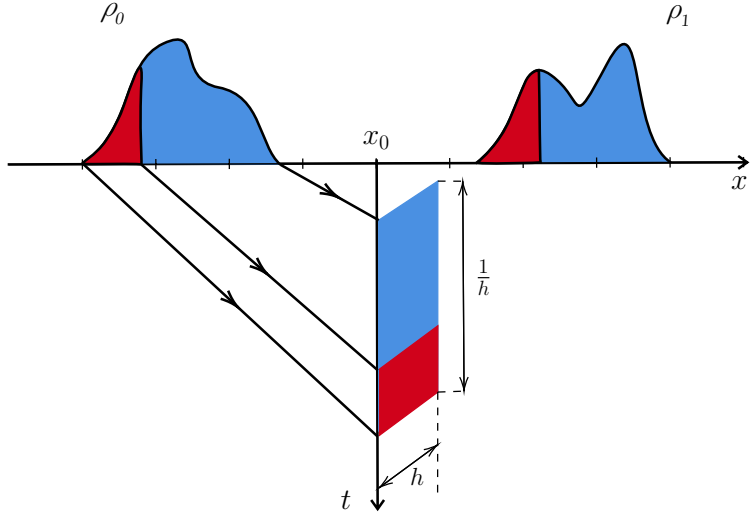

\section{Problem formulation}\label{sec:sec1}
We consider two probability densities $\rho_0, \rho_1$ on $\mathbb{R}$ having finite second-order moments.
For $Y:[0,1]\times \mathbb{R} \rightarrow \mathbb{R}$ such that\footnote{As is common, $Y_{t\# \rho_0}$ denotes the push-forward of $\rho_0$ under $Y_t$, see \cite{villani2009optimal}.} $Y_{0\# \rho_0} = \rho_0$ and  $Y_{1\# \rho_0} = \rho_1$, we are interested in minimizing 
\begin{equation}\label{eq:Jv} 
J(\partial_t Y) := \int_0^1 \int_\mathbb{R} (\partial_t Y_t(x))^2\rho_0(x)dxdt.
\end{equation}
In the absence of any  additional  constraint on $Y$, the solution is $Y^\star_t(x)= x + t(T(x)-x)$ for $T$ the optimal transport map between $\rho_0$ and $\rho_1$ and $J(\partial_tY_t^\star)=\mathcal{W}_2^2(\rho_0,\rho_1)$, 
the squared Wasserstein-2 distance between the two \cite{villani2009optimal}. 
Here however, for a certain $x_0 \in \mathbb{R}$, we introduce a constraint on the flux passing through $x_0$, as explained below. Throughout the paper, $T$ will always denote the optimal transportation plan in the absence of any such constraint. The purpose of the present work of course is to develop theory that addresses the case of transport with a bound on the flux through $x_0$.

When all functions are smooth and well defined, a flux constraint at $x_0$ can be expressed as 
$$|\rho_t(x_0)v_t(x_0)|\leq h\ \ \ \forall t\in (0,1)$$ for $\rho_t$ the density of $Y_{t\# \rho_0}$ and $v_t(x_0)=\partial_t Y_t(Y_t^{-1}(x_0))$. However in the general case, if $\rho_t$ is not continuous (or doesn't even exist), this constraint is not well defined.
One way to deal with such a situation is to recast the constraint as requiring that\footnote{We use the standard notation $\mathds{1}_A$ for the characteristic function of the set $A$.}, $\forall t\in (0,1)$,
\begin{equation}\label{eq:c1}\limsup \limits_{\substack{\alpha_1 \rightarrow 0 \\ \alpha_2 \rightarrow 0}}\ \frac{1}{|\alpha_2-\alpha_1|} \int \mathds{1}_{\{x_0 \in (Y_t(x) + \alpha_1, Y_t(x) + \alpha_2)\}} |\partial_t Y_t(x)|\rho_0(x)dx \leq h. \ \end{equation}
Then, if $\rho_0$ is continuous and $Y_t$ is a $C^1$ diffeomorphism, the left hand side (LHS) of \eqref{eq:c1}  amounts to
\begin{align*}
\mbox{LHS } \eqref{eq:c1} & = \limsup \limits_{\substack{\alpha_1 \rightarrow 0 \\ \alpha_2 \rightarrow 0}}\frac{1}{|\alpha_2-\alpha_1|}\int \mathds{1}_{\{x_0 \in (y + \alpha_1, y + \alpha_2)\}} |\partial_t Y_t(Y_t^{-1}(y))|\rho_t(y)dy \\
& = \rho_t(x_0)v_t(x_0),
\end{align*}

Interestingly, when $Y_t$ fails to be a $C^1$ diffeomorphism, special care is needed. For instance, take $x_0=0$ and $Y_t(x) = \mathds{1}_{\{x\in [-2,-1]\}}(1-2t)^3x$. The constraint \eqref{eq:c1} is satisfied since $\partial_t Y_t(x)=0$ at $t=1/2$, and no mass sits near the toll for any $t\neq 1/2$. Thus, the formulation \eqref{eq:c1} fails to capture the situation where infinite mass passes through with zero velocity. We reformulate so as to avoid this technicality. 

Consider the modified constraint that bounds the flux passing through $x_0$, expressed as requiring that $\forall t\in (0,1)$
\begin{equation}\label{eq:c2}
\limsup \limits_{\substack{\alpha_1 \rightarrow 0 \\ \alpha_2 \rightarrow 0}}\ \frac{1}{|\alpha_2-\alpha_1|} \int \mathds{1}_{\{x_0 \in (Y_{t + \alpha_1}(x), Y_{t + \alpha_2}(x))\}} \rho_0(x)dx
\leq h.
\end{equation}
In the case where $Y_t$ is $C^1$, using the Taylor expansion of $Y$ in time, the left hand side (LHS) of \eqref{eq:c2}  amounts to
\begin{align}
\mbox{LHS } \eqref{eq:c2}  = & \limsup \limits_{\substack{\alpha_1 \rightarrow 0 \\ \alpha_2 \rightarrow 0}} \ \frac{1}{|\alpha_2-\alpha_1|} \int \mathds{1}_{\{x_0 \in (Y_{t}(x)+\partial_tY_{t }(x)\alpha_1+o(\alpha_1), Y_{t}(x)+\partial_tY_{t }(x)\alpha_2+o(\alpha_2))\}} \rho_0(x)dx\nonumber\\
  =  & \limsup \limits_{\substack{\alpha_1 \rightarrow 0 \\ \alpha_2 \rightarrow 0}}
\int \hspace*{-5pt} \Big(
\mathds{1}_{\{Y_t(x) \in (x_0-\partial_tY_{t }(x)\alpha_2+o(\alpha_2), x_0-\partial_tY_{t }(x)\alpha_1+o(\alpha_1))\}} \frac{|\partial_t Y_t(x)|\mathds{1}_{\{|\partial_t Y_t(x)|>0\}}}{|\alpha_2-\alpha_1||\partial_t Y_t(x)|}\nonumber\\
  & +\mathds{1}
 _{\{Y_t(x) \in (x_0+o(\alpha_1), x_0+o(\alpha_2))\}} \frac{\mathds{1}_{\{\partial_t Y_t(x)=0\}}}
 {|\alpha_2-\alpha_1|}
 \Big) \rho_0(x)dx .\nonumber
 \end{align} 
Using a change of variables, we readily see that \eqref{eq:c2} implies \eqref{eq:c1} and that if $Y_t$ is a $C^1$ diffeomorphism, the two constraints are identical. Note also that, $\forall t\in (0,1)$, condition \eqref{eq:c2} is equivalent to 
\begin{equation}\label{eq:c3}
\forall \alpha_1,\alpha_2\in \mathbb{R}, \quad   \int \mathds{1}_{\{x_0 \in (Y_{t + \alpha_1}(x), Y_{t + \alpha_2}(x))\}} \rho_0(x)dx \leq h|\alpha_2-\alpha_1|.
\end{equation}

Define $\Omega = \{x\in \supp(\rho_0)\ | \ x_0 \in (x,T(x)) \text{ or } x_0 \in (T(x),x)\}$, where $T$ is the optimal transport map of the unconstrained problem. Thus, $\Omega$ contains the support of mass that needs to cross the toll station, at some point in time, in either direction.
From \eqref{eq:c3} it is evident that $h \geq \rho_0(\Omega)$ is necessary for the existence of a map satisfying the constraint (since the transport will take place over the time interval $[0,1]$). Typically, $h > \rho_0(\Omega)$ is required, except in some special cases where $h = \rho_0(\Omega)$ may suffice, as for example when $\rho_0 =\mathds{1}_{\{[0,1]\}}$, $x_0=1$ and $\rho_1 =\mathds{1}_{\{[1,2]\}}$. From here on we assume that $h> \rho_0(\Omega)$.

We are now in a position to cast our optimization problem in terms of a velocity field $v_t(x)$ that will effect the transport; formally, $v_t(x)=\partial_tY_t(x)$ relates to our earlier notation when functions are smooth. For any $v \in L^2([0,1]\times \mathbb{R}, \mathbb{R})$, define the map $Y^v:[0,1]\times \mathbb{R}\rightarrow \mathbb{R}$ as the flow of $v$:
\begin{equation}\label{eq:Yintegral}
Y^v_t=\Id+\int_0^t v_\tau d\tau,
\end{equation}
with $\Id$ denoting the identity map in $\mathbb{R}$. Our problem can now be stated as follows. 

\begin{problem}\label{prob:problem1}
 Consider
\begin{equation}\label{prob2} 
\inf \limits_{v\in \Lambda} J(v),
\end{equation}
over the class $\Lambda$ of functions $v \in L^2([0,1]\times \mathbb{R}, \mathbb{R})$ defined so that $Y^v$, the flow of v, satisfies
\begin{itemize}
    \item[i)] $Y_{1\# \rho_0}^v = \rho_1$
    \item[ii)] $\forall t \in (0,1)$, $Y_t^v$ satisfies the constraint \eqref{eq:c3}.
\end{itemize}
Determine existence, uniqueness, and a functional form for a minimizing solution $v$.
\end{problem}

\blue
\section{Existence of a solution}\label{sec:sec2}
\black
We say that a map $Y:[0,1]\times \mathbb{R}\rightarrow \mathbb{R}$ is in the set $\Gamma$ if there exist $v\in \Lambda$ such that $Y$ is the flow of $v$, i.e. $Y_t=Y^v_t=\Id+\int_0^t v_\tau d\tau$. From here on, the $v$ in the notation $Y^v_t$ is suppressed as we are truly interested in the transport map. 
We first derive certain useful
properties of candidate minimizers of our problem.
To this end, for any $Y \in \Gamma$ and $x\in \Omega$, we define
\[
\toll_Y(x)=\inf \{ t \mid x_0=Y_t(x)\}.
\]
Thus, the function $\toll_Y$ specifies the times of transit through the toll station of mass that is initially located at $x$ and then transported via $Y$.

It is clear that the function $\toll_Y$ must be injective\footnote{This follows by cyclic monotonicity since the cost is convex, see \cite[Section 2.3]{villani2021topics}.} for a minimizing solution, and that mass flow takes place always in the same direction across the toll station. Then, $\forall t\in (0,1)$, \eqref{eq:c2} is equivalent to 
\[
\limsup \limits_{\substack{\alpha_1 \rightarrow 0 \\ \alpha_2 \rightarrow 0}}\ \frac{1}{|\alpha_2-\alpha_1|}\toll_{Y\#\rho_0}((t+ \alpha_1 ,t+ \alpha_2)) \leq h,
\]
and so, if $\toll_{Y\#\rho_0}$ (the measure on $[0,1]$ that weighs the mass that goes through $x_0$ at different times $t\in[0,1]$) admits a continuous density $\varrho_{\toll}$, the constraint amounts to $\varrho_{\toll}(t)\leq h$. Note also that this condition is different than simply stating  $\rho_t(x_0)\leq h$, as the latter doesn't take into account the speed of transport. Then we see that for $x\notin \Omega$, we can restrict ourselves to considering maps $Y\in \Gamma$ such that $Y_t(x)=x + t(T(x)-x)$ for $T$ the optimal transport map between $\rho_0$ and $\rho_1$. Thus, in the sequel, without loss of generality we always suppose that $\Omega=\supp(\rho_0)$ and that $\sup \supp(\rho_0)\leq x_0 \leq \inf \supp(\rho_1)$.

For $Y\in \Gamma$ and $\toll_Y$ its corresponding transit-time function, define the map $\overline{Y}:[0,1]\times \mathbb{R}\rightarrow \mathbb{R}$ by 
\begin{equation}\label{eq:form}
\overline{Y}_t(x) = \left\{
\begin{array}{ll}
 x+t\frac{x_0-x}{\toll_{Y}(x)} & \text{if } t\leq \toll_{Y}(x) \medskip\\
    x_0+(t-\toll_{Y}(x))\frac{T(x)-x_0}{1-\toll_{Y}(x)} & \text{if } t \geq \toll_{Y}(x) 
\end{array}
\right.
\end{equation} 
and note $\overline{\Gamma}=\{ \overline{Y} \ |\ Y \in \Gamma\}$ the set of functions of this type. The next statement states that we can restrict our minimization problem to functions of the form \eqref{eq:form}.  Specifically, it states that for any candidate minimizer $Y \in \Gamma$, the speed of transport needs to remain constant at all times prior to transit, and again, constant at all times after transit. In addition, from the functional form, we see that $Y_1=T(x)$ for all $x$. This last statement says that the final destination of mass originally located at $x$ is the same, whether we apply $T$ or the optimal plan that abides by the constraint; the only thing that changes in the two cases is the speed while the mass traverses the segment before $x_0$ and after 
(cf.\ example in Section \ref{sec:example}).
\begin{proposition}\label{prop:prop1} We have $$\inf \limits_{Y\in \Gamma} J(\partial_t Y) = \inf \limits_{\overline{Y}\in \overline{\Gamma}} J(\partial_t \overline{Y})$$

\end{proposition}

\begin{proof}
For $Y\in \Gamma$ and $\toll_Y$, define
\begin{equation} 
\label{bar}
Y^c(x) = \left\{
\begin{array}{ll}
 x+t\frac{x_0-x}{\toll_Y(x)} & \text{if } t\leq \toll_Y(x) \medskip\\
    x_0+(t-\toll_Y(x))\frac{Y_1(x)-x_0}{1-\toll_Y(x)} & \text{if } t \geq \toll_Y(x) 
\end{array} 
\right.
\end{equation}
Thus, $Y^c$ maintains the terminal destination $Y_1(x)$ and the crossing time $\toll_Y(x)$, for the mass that was initially at $x$, while it ensures constancy of speed before and after crossing.
It follows that $Y^c \in \Gamma$ and that $J(\partial_t Y^c)\leq J(\partial Y)$, by convexity, so we can restrict $\mathcal Y$ to the set of functions that are of the form \eqref{bar} since candidate minimizers will always be of that form.

As the position $Y_1(x)$ in \eqref{bar} doesn't impact  the constraint \eqref{eq:c2}, we consider how $Y_1(x)$ may depend on the time of crossing $\toll_Y(x)$. Specifically,
$Y_1$ must be a minimum for the cost
\[
 \int_{\toll_Y(x)}^1 \int_\mathbb{R} \left(\frac{Y_1(x)-x_0}{1-\toll_Y(x)}\right)^2\rho_0(x)dxdt = \int_\mathbb{R} \frac{(Y_1(x)-x_0)^2}{1-\toll_Y(x)}\rho_0(x)dx.
\]
From this we deduce that
$\toll_Y(x)\leq \toll_Y(y)$ iff $Y_1(x)\geq Y_1(y)$. Furthermore, as the problem is reversible (we can switch $\rho_0$ and $\rho_1$), we can deduce in the same way that $\toll_Y(x)\geq \toll_Y(y)$ iff $x\leq y$. Therefore $Y_1(x)$ is increasing and we conclude that it is identical to $T$ the optimal transport map between $\rho_0$ and $\rho_1$.
\end{proof}

From Proposition \ref{prop:prop1} we also deduce that for $Y$, the flow of a (candidate) optimal solution, the map $x\mapsto \toll_{Y}$ is strictly decreasing 
on the support of $\rho_0$,
and that $Y_t$ is one to one, for all $t$. 

Let us write $\f(x)=\frac{x_0-x}{\toll_Y(x)}$ for the velocity of transport {\em prior} to crossing the toll, for the mass initially located at $x$ at the start. Then, in light of Proposition \ref{prop:prop1}, our problem is reduced to finding
\begin{align}\f\in \argmin & \int_0^1 \int_\mathbb{R} \Big( \f(x)^2\mathds{1}_{\{t\leq \frac{x_0-x}{\f(x)}\}} + \left(\frac{T(x)-x_0}{1-\frac{x_0-x}{\f(x)}}\right)^2\mathds{1}_{\{t\geq \frac{x_0-x}{\f(x)}\}}\Big) \rho_0(x)dxdt\nonumber \\\label{prob3}
= & \int_\mathbb{R} \Big( \f(x)(x_0-x) + \frac{(T(x)-x_0)^2}{1-\frac{x_0-x}{\f(x)}}\Big) \rho_0(x)dx,
\end{align} 
subject to $x\mapsto \frac{x_0-x}{\f(x)}=\toll_\f(x)$ being decreasing and bounded between 0 and 1, 
and
\begin{equation}\label{eq:c4}\limsup \limits_{\substack{\alpha_1 \rightarrow 0 \\ \alpha_2 \rightarrow 0}}\ \frac{1}{|\alpha_2-\alpha_1|} \int \mathds{1}_{\{(t+\alpha_1)\f(x)< x_0-x < (t+\alpha_2)\f(x)\}} \rho_0(x)dx \leq h. \end{equation}
We now argue the existence of a minimizer $\f^\star$.

{\blue
\begin{proposition}
Supposing that the two probabilities densities $\rho_0,\rho_1$ have finite second-order moments, Problem \ref{prob:problem1} admits a solution.
\end{proposition}
}

\begin{proof}
Let  $(\f_n)_n$ be a minimizing sequence of \eqref{prob3} and write $toll_n:\supp(\rho_0)\rightarrow (0,1)$ the associated toll function: $toll_n(x)=\frac{x_0-x}{\f_n(x)}$. Let $(\alpha_k)_k$ be a dense sequence in $\supp(\rho_0)$ (for example the rational numbers). By compactness, we have that $\forall k\in \mathbb{N}$, $toll_n(a_k)$ admits a converging subsequence in $n$. Then using a diagonal argument, 
there exist a subsequence $(\f_{\varphi(n)})_n$ and $\beta_k\in [0,1]$ such that, $\forall k\in \mathbb{N} $,
$toll_{\varphi(n)}(a_k)\xrightarrow[n \to +\infty]{} \beta_k$ and $\alpha_k\leq \alpha_l \iff  \beta_k\leq \beta_l$. For $x\in \supp(\rho_0)$,
and $(\alpha_{\psi(k)})_k$ a decreasing subsequence converging to $x$, let be $toll(x)=\lim_k \beta_{\psi(k)}$, which is well defined as $\beta_{\psi(k)}$ is decreasing. 
Then $toll_{\varphi(n)}(x)$ converges to $toll(x)$ for any $x$ being a point of continuity of $toll$. 
As $toll$ is a nonincreasing map, it has at most a countable number of points of discontinuity, therefore $toll_{\varphi(n)}$ converges to $toll$ a.e. In particular we get that $toll_{\varphi(n)\# \rho_0}$ converges weakly to $toll_{\# \rho_0}$. For $x\in \supp(\rho_0)\setminus \{x\ | \ toll(x)=0\}$, define $\f(x)=\frac{ x_0-x}{toll(x)}$, it is well defined a.e. because $\{x\ | \ toll(x)=0\}$ has measure 0 as $(\f_n)_n$ is a minimizing sequence. Then $\f_{\varphi(n)}$ converges a.e. to $\f$  and as the constraint \eqref{eq:c4} is equivalent to 
\[
\forall \alpha_1,\alpha_2\in \mathbb{R}, \quad   \ \toll_{\f}((t+ \alpha_1 ,t+ \alpha_2)) \leq h|\alpha_2-\alpha_1|,
\]
$\f$ verifies the constraint. Finally, by lower semi continuity of the cost, $\f$ is a minimizer of \eqref{prob3}. \end{proof}


\blue
\section{Uniqueness of the solution}\label{sec:sec2bis}
\black
Before we proceed with the proof of uniqueness of the minimizer, we recast
our problem in terms of flux as the optimization variable. For $u \in L^1([0,1]\times \mathbb{R},\mathbb{R})$, a candidate flux (i.e., mass times velocity), define a corresponding mass-measure $\rho_t^u$ on $\mathbb{R}$ by duality via: $\forall \phi \in C_c^\infty(\mathbb{R},\mathbb{R})$,\
$$
\int_\mathbb{R} \phi(x)d\rho_t^u(x) = \int _\mathbb{R} \phi(x)\rho_0(x)dx+\int_0^t\int_\mathbb{R} (\nabla \phi(x))u_r(x)dxdr .
$$
Equivalently, we have that $\rho^u$ solves in the weak sense the continuity equation 
$$
\left\{\begin{array}{ll}
 \partial_t\rho_t^u = - \nabla \cdot u\\
    \rho_0^u=\rho_0
\end{array}\right..
$$
For a flux $u$ such that $\forall t \in (0,1)$, $\rho^u_t$ admits a positive density, let us express the cost of $u$ as
\newcommand{\J}{{\rm J}}
\begin{equation}\label{eq:convexcost}
\J(u)= \int_0^1 \int_\mathbb{R} \frac{u_t(y)^2}{\rho_t^u(y)}dydt
\end{equation}
In the above, by a slight abuse of notation as it is often done, we used $\rho^u$ to denote both the measure and the corresponding density, allowing these to be distinguished by the specific usage and context.

\begin{problem}\label{prob:problem2}
 Consider
\begin{equation}\label{prob3a}
\inf \limits_{u\in {\mathcal U}} \J(u). 
\end{equation} 
over the class ${\mathcal U}$ defined as the set of functions $u \in L^1([0,1]\times \mathbb{R},\mathbb{R})$ a.e.\ such that
\begin{itemize}
    \item[i')] $\forall t \in (0,1)$, $\rho^u_t$ admits a positive density and $\rho^u_1 = \rho_1$ 
    \item[ii')] satisfy \begin{equation}\label{eq:c5_earlier}
    \forall t \in (0,1), \limsup \limits_{x_1 \rightarrow x_0,\ x_2 \rightarrow x_0}\ \frac{1}{|x_2-x_1|} \int_{x_1}^{x_2} |u_t(y)|dy \leq h.\end{equation} 
\end{itemize}
Determine existence, uniqueness, and a functional form for a minimizing solution $u$.
\end{problem}

We will first
prove the equivalence of the above formulation in Problem \ref{prob:problem2} with that in Problem
\ref{prob:problem1}.
The advantage of Problem \ref{prob:problem2} is that the constraint is now convex which will be convenient in proving
uniqueness. Note that here we use roman $\J$ with argument the flux field, to echo the earlier usage in \eqref{eq:Jv} where the action integral $J$ first appeared with argument the velocity.

{\blue
\begin{proposition}
\label{prop:propequivalence}
Problems \ref{prob:problem1} and \ref{prob:problem2} are equivalent.
\end{proposition}
}

\begin{proof}
Let $Y\in \overline{\Gamma}$ be a solution of Problem \ref{prob:problem1}, $v_t(\cdot)=\partial_tY_t(Y_t^{-1}(\cdot))$ the associated velocity (defined everywhere except at the points $(toll(x),x)$, for all $x\in \supp(\rho_0)$)  and $\rho_t=Y_{t\#\rho_0}$ the associated mass flow. Then for $\phi \in C_c^\infty(\mathbb{R},\mathbb{R})$ we have 
\begin{align*}
\int_\mathbb{R} \phi(x)d\rho_t(x) & = \int_\mathbb{R} \phi(Y_t(x))\rho_0(x)dx\\
& = \int_\mathbb{R} \Big(\phi(Y_0(x))+\int_0^t \partial_t\phi(Y_r(x))dr\Big)\rho_0(x)dx\\
& = \int_\mathbb{R} \Big(\phi(Y_0(x))+\int_0^t \nabla \phi(Y_r(x)) v_r(Y_r(x))dr\Big)\rho_0(x)dx
\\
& = \int_\mathbb{R} \phi(x) \rho_0(x)dx +\int_0^t\int_\mathbb{R} \nabla \phi(x) v_r(x)d\rho_r(x)dr.
\end{align*}
Therefore $Y$ defines a unique flux $u\in L^1([0,1]\times \mathbb{R},\mathbb{R})$ ($u$ is $L^1$ by Jensen inequality) by $u_t(x)=v_t(x)\rho_t(x)$ with $\J(u)=J(\partial_t Y)$. Furthermore, for $\f(x)=\partial_tY_0(x)=\frac{x_0-x}{\toll_\f(x)}$ we have that the left hand side of \eqref{eq:c4} amounts to
\begin{align*}
   \mbox{LHS }\eqref{eq:c4} & = \limsup \limits_{\substack{\alpha_1 \rightarrow 0 \\ \alpha_2 \rightarrow 0}}\  \int \mathds{1}_{\{Y_t(x) +\alpha_1\f(x)< x_0< Y_t(x)+\alpha_2\f(x)\}} \frac{\f(x)}{|\alpha_2-\alpha_1|\f(x)}\rho_0(x)dx\\
    & = \limsup \limits_{\substack{\epsilon_1 \rightarrow 0 \\ \epsilon_2 \rightarrow 0}}\  \frac{1}{|\epsilon_2-\epsilon_1|}\int \mathds{1}_{\{y +\epsilon_1< x_0< y+\epsilon_2\}} v_t(y)\rho_t(y)dy.
\end{align*}
Therefore $u\in {\mathcal U}$ and we conclude that $\inf \limits_{u\in {\mathcal U}} \J(u)\leq \min \limits_{Y \in \Gamma} J(\partial_t Y)$.

For establishing the reverse direction, let $u\in {\mathcal U} \cap C([0,1],C^1_c(\mathbb{R},\mathbb{R}))$ with $\J(u)< \infty$ and define $T_t^u$ the optimal transport map between $\rho_0$ and $\rho_t^u$. For $F_t(x)=\int_{-\infty}^x \rho_t^u(x) dx$ the cumulative distribution function of $\rho_t^u$, it is well known that $T_t^u(x)=F_t^{-1}(F_0(x))$, see \cite[Chapter 1]{villani2021topics}. Since $\forall t\in [0,1]$ we have  $F_t(F_t^{-1}(F_0(x)))=F_0(x)$, differentiating this expression we have
\begin{align*}
    \partial_t F\bigg|_{(t,x)=(t,F_t^{-1}(F_0(x)))}+ \partial_x F\bigg|_{(t,x)=(t,F_t^{-1}(F_0(x)))} \partial_t F_t^{-1}(F_0(x)) & = 0\\
    \int_{-\infty}^{F_t^{-1}(F_0(x))} \partial_t d\rho^u_t(x) + \rho^u_t(F_t^{-1}(F_0(x)))\partial_t F_t^{-1}(F_0(x)) & = 0\\
    \int_{-\infty}^{T_t^u(x)} -\nabla u_t(z)dz + \rho^u_t(T_t^u(x))\partial_t T_t^u(x) & = 0\\
    \Rightarrow \partial_t T_t^u(x) & = \frac{u_t(T_t^u(x))}{\rho^u_t(T_t^u(x))}. 
\end{align*}
Therefore $T_t^u$ defines a map in $\Gamma$ such that $\J(u)=J(\partial_t T_t^u)$. Then, since the space $C([0,1],C^1_c(\mathbb{R},\mathbb{R}))$ is dense in $L^1([0,1]\times\mathbb{R},\mathbb{R})$, we deduce that $\min \limits_{u\in {\mathcal U}} \J(u)= \min \limits_{Y \in \Gamma} J(\partial_t Y)$.\end{proof}

{\blue
Using the equivalence of Problem \ref{prob:problem1} and Problem \ref{prob:problem2}, we can now prove the uniqueness of the minimizer.
\begin{theorem}\label{thm:thm1}
Problem \ref{prob:problem2} (and so Problem \ref{prob:problem1}) admits a unique solution.
\end{theorem}
}

\begin{proof}
Suppose that we have $u_1$ and $u_2$, two solutions of \eqref{prob3}. For $\lambda\in (0,1)$, 
by convexity we have that $\frac{(\lambda u_1+(1-\lambda)u_2)^2}{\lambda\rho^{u_1}+(1-\lambda)\rho^{u_2}}\leq \lambda\frac{u_1^2}{\rho^{u_1}}+(1-\lambda)\frac{u_2^2}{\rho^{u_2}}$, 
but as they are both solutions, this is an equality. 
However the polynomial $\lambda \mapsto
(\lambda (u_1-u_2)+u_2)^2
-(\lambda(\rho^{u_1}-\rho^{u_2})+\rho^{u_2}) (\lambda(\frac{u_1^2}{\rho^{u_1}}-\frac{u_2^2}{\rho^{u_2}})+\frac{u_2^2}{\rho^{u_2}})$ 
is identically zero iff $\frac{u_1}{\rho_1}=\frac{u_2}{\rho_2}$, and iff $v_1=v_2$.\end{proof}


{\blue
\section{Properties and structural form of the solution under smoothness assumption}\label{sec:sec3}
We are now in a position to build explicitly the solution $\f^\star$ of Problem \eqref{prob3} in the case when $\rho_0$ and $\rho_1$ have additional smoothness assumptions. All along this section, we will assume that $\rho_0$ and $\rho_1$ are continuous, have bounded convex support, and are bounded from below on the interior of their support. In the process of building the solution, we also establish structural properties of the solution.
}

Under the stated assumptions on $\rho_0,\rho_1$, by using the closed-form expression for the optimal transport map $T$ in dimension one \cite[Chapter 1]{villani2021topics}, it is immediate to see that $T$ is $C^1$.

Recall first that, without loss of generality, we assume that $\sup \supp(\rho_0)\leq x_0\leq \inf \supp(\rho_1)$.
For $\f:\supp(\rho_0)\rightarrow \mathbb{R}$ such that\footnote{The notation $\toll_\f$ signifies $\toll_Y$, for the corresponding $Y$ obtained via \eqref{eq:Yintegral}.} $x\mapsto \toll_\f(x)= \frac{x_0-x}{\f(x)}$ is decreasing and bounded between 0 and 1 on $\supp(\rho_0)$, the expression 
$$C_y(\f) = \limsup \limits_{\substack{\alpha_1 \rightarrow 0 \\ \alpha_2 \rightarrow 0}}\ \frac{1}{|\alpha_2-\alpha_1|} \int \mathds{1}_{\{(\toll_\f(y)+\alpha_1)\f(x)< x_0-x < (\toll_\f(y)+\alpha_2)\f(x)\}} \rho_0(x)dx$$
gives the value of the flux passing through the toll station when the mass initially at $y$ is crossing. \blue Let first prove that from the additional assumptions on $\rho_0$ and $\rho_1$, we have that the solution is continuous.

\begin{proposition}
The solution $\f^\star \in L^2$ admits a continuous representative.
\end{proposition}

\begin{proof}
From section 1, we know that the solution $\f^\star \in L^2$ admits a representative such that the function $x\mapsto \toll_{\f^\star}(x)= \frac{x_0-x}{\f^\star(x)}$ is decreasing. Now by absurd, suppose that $\f^\star$ is not continuous. Then there exists $x_0\in \supp(\rho_0)$ and $\epsilon>0$ such that $\forall \delta>0$, $\exists x_\delta\in \supp(\rho_0) $ with $|x_0-x_\delta|<\delta$ and $|\f^\star(x_0)-\f^\star(x_\delta)|>\epsilon$. As $\toll_{\f^\star}$ is decreasing, we have that for $\delta$ small enough, 
$$\frac{\f^\star(x_0)-\f^\star(x_\delta)}{|\f^\star(x_0)-\f^\star(x_\delta)|} = \frac{x_0-x_\delta}{|x_0-x_\delta|}
$$
so $\toll_{\f^\star}$ is not continuous in $x_0$ neither. Suppose now that $\forall \delta>0$, $x_\delta-x_0>0$ (the proof would be the same for $x_\delta-x_0<0$). Then we have that
$$
\lim \limits_{\substack{x \rightarrow x_0 \\ x>x_0}}\ \toll_{\f^\star}(x)< \toll_{\f^\star}(x_0).
$$
If $\f^\star(x_0)<T(x_0)-x_0$, then as $\toll_{\f^\star}$ is decreasing and $T$ is continuous, we have that for $\gamma>0$ small enough, $\f^\star(x)+2\gamma \leq T(x)-x$ for all $ x \in (x_0-\gamma,x_0]$. Then by strict convexity of $J$, the function 
$$
\f_2(x)=\f^\star(x)+\gamma \mathds{1}_{\{x \in (x_0-\gamma,x_0]\}}
$$ verifies that $J(v_2)<J(\f^\star)$. Furthermore for $\gamma$ small enough, we have that $C_x(v_2)<h$ for all $x\in (x_0-\gamma,x_0)$, as $\rho_0$ is continuous and $\toll_{\f^\star}$ is decreasing so $C_x(\f^\star)<C_{x_\delta}(\f^\star)$ for $\delta$ small enough. Therefore we have that $v_2$ is a better solution to the problem.\\
If $\f^\star(x_0)\geq T(x_0)-x_0$, then by continuity of $T$ we have that for $\gamma>0$ small enough, $\f^\star(x)+2 \gamma \geq T(x)-x$, for all $ x \in (x_0,x_0+\gamma]$. As previously we can find a better solution $\f_2(x)=\f^\star(x)-\gamma \mathds{1}_{\{x \in (x_0,x_0+\gamma]\}}$ to the problem which contradicts the fact that $\f^\star$ is the minimizer.
\end{proof}
\black

The next proposition states that at the points where $\rm \f^\star$ doesn't saturate the constraint, $\rm \f^\star$ is equal to the unconstrained transport $T-\Id$.

\begin{proposition}\label{prop2}
If there exist $y\in \supp(\rho_0)$ such that $C_y(\f^\star)  < h$, then we have $\f^\star(y)=T(y)-y$.
\end{proposition}

\begin{proof}
Suppose $\exists y \in \supp(\rho_0)$ such that $C_y(\f^\star)  < h$ and $\f^\star(y)\neq T(y)-y$. Define 
$g_\epsilon(x)=\mathds{1}_{\{x \in (y-\epsilon,y+\epsilon)\}}\epsilon^3\exp(-\frac{1}{\epsilon^2-(x-y)^2}+\frac{1}{\epsilon^2})$. Then there exist $\epsilon\neq 0 \in \mathbb{R}$ and $\delta>0$ such that $\forall x \in (y-\delta,y+\delta)$ we have $|\f^\star(x)+g_\epsilon(x)-(T(x)-x)|<|\f^\star(x)-(T(x)-x)|$ and $C_x(\f^\star+g_\epsilon) < h$, since $g_\epsilon$ introduces a vanishingly small bump at a suitable location. By strict convexity of $J$ we have that $J(\f^\star)>J(\f^\star+g_\epsilon)$ which contradicts the optimality of $\f^\star$.
\end{proof}

We can now deduce some regularity of the function $\f^\star$.
\begin{corrolary}
The optimal solution $\f^\star$ of \eqref{prob3} is $C^1$ almost everywhere.
\end{corrolary}

\begin{proof}
\blue As $T$ is $C^1$, then  $\f^\star$ is also $C^1$ at points $y$ that lie in the interior of the closed set $\{y\in \supp(\rho_0)\ |\ \f^\star(y)=T(y)-y\}$. \black  Otherwise if for some $y$ it holds that $\f^\star(y)\neq T(y)-y$, then $\exists \delta>0$ such that $\forall x \in (y-\delta,y+\delta)$, $\f^\star(x)\neq T(x)-x$ which implies by Proposition \ref{prop2} that $C_x(\f^\star) = h$. Solve the ordinary differential equation
\begin{equation}\label{eq:quadratic}
\left\{\begin{array}{ll}
  \partial_x\f(x)=\frac{\f(x)^2\rho_0(x)-h\f(x)}{h(x_0-x)} \quad \quad \text{ for } y-\delta \leq x\leq y_1\\
    \f(y+\delta)=\f^\star(y+\delta)
\end{array}\right.
\end{equation}
for 
$\f(x)$. 
It can be shown that the function $\f$ is well defined by establishing existence and uniqueness of the solution to \eqref{eq:quadratic} using the Cauchy-Lipschitz theorem and inherent boundedness.
Indeed, if 
\[
\f(x)> \frac{h}
{\inf\{\rho_0(y)\mid y\in \supp(\rho_0)\}},
\]
then $\partial_x \f(x)>0$, and so $\f$ is decreasing with decreasing value of its argument on a small interval $[x-\epsilon,x]$, and if
\begin{equation*}
    0<\f(x)<\frac{h}{\sup \{\rho_0(y)\mid y\in \supp(\rho_0)\}},
\end{equation*}
then $\partial_x \f(x)<0$, and so $\f$ is increasing (again with decreasing value of its argument) on a small interval $[x-\epsilon,x]$. As $\f^\star(y+\delta)>0$, and $\f\mapsto \frac{\f^2\rho_0(x)-h\f}{h(x_0-x)}$ is Lipschitz on any compact set, we can apply the Cauchy-Lipschitz theorem to establish existence and uniqueness. 
From the definition of $\f$, it follows
that $C_x(\f)=h$, and therefore $\f$ has the same flux as $\f^\star$. By uniqueness, $\f$ which is $C^1$ on $[y-\delta,y_1]$, is optimal, i.e.,
$\f=\f^*$.
\blue Finally as $\f^\star$ is $C^1$ on the interior of the set $\{y\in \supp(\rho_0)\ |\ \f^\star(y)=T(y)-y\}$ and is also $C^1$ on the set $\{y\in \supp(\rho_0)\ |\ \f^\star(y)\neq T(y)-y\}$, we deduce that $\f^\star$ is $C^1$ almost everywhere as the boundary of those two sets is at most countable. \black
\end{proof}

Now that we have established that $\f^\star$ is $C^1$ a.e., we can write the constraint \eqref{eq:c4} for functions $\f\in C^1(\supp(\rho_0),\mathbb{R})$ as: for $x\in \supp(\rho_0)$, a.e.
\begin{equation}\label{eq:c5} C_x(\f)=\frac{\f(x)\rho_0(x)}{1+\frac{x_0-x}{\f(x)}\partial_x\f(x)}\leq h.
\end{equation}
For $\f:\supp(\rho_0)\rightarrow \mathbb{R}$, define 
$$
J(\f)=\int_\mathbb{R} \Big( \f(x)(x_0-x) + \frac{(T(x)-x_0)^2}{1-\frac{x_0-x}{\f(x)}}\Big) \rho_0(x)dx.
$$
We can then rewrite Problem \ref{prob:problem1} in the present case where $\rho_0$ and $\rho_1$ are continuous,
have bounded convex support and are bounded from below on the interior of  their support, as follows.

\begin{problem}\label{prob:problem3} Consider 
\begin{equation}
    \min \limits_{\f\in V} J(\f)
\end{equation}
over a class $V$ of functions $\f:\supp(\rho_0)\rightarrow \mathbb{R}$, that are $C^1$ a.e.\ and are such that
\begin{itemize}
    \item[i)] the map $x\mapsto \frac{x_0-x}{\f(x)}$ is decreasing and bounded between 0 and 1
    \item[ii)] $\f$ verifies condition \eqref{eq:c5} a.e.
\end{itemize}
\end{problem}

To solve Problem \ref{prob:problem3}, we define velocity fields $\f$
on all of $\mathbb{R}$, even outside $\supp(\rho_0)$, as this suitably defined {\em prolongation} of $\f$ will be conveniently expressed as a solution of a differential equation. To this end, we note that the constraint \eqref{eq:c5} can be alternatively expressed in the form
\begin{equation}\label{eq:forrho0}
C_x^{\rm alt}({\f}):=\frac{\f(x)^2\rho_0(x)-h(x_0-x)\partial_x\f(x)}{\f(x)}\leq h.
\end{equation} 
This alternative formulation applies even for points $x$ where $\rho_0=0$, and will help define the sought prolongation for $\f^\star$.

Let us first prolong on all of $\mathbb{R}$ the optimal transport map between $\rho_0$ and $\rho_1$.
To this end, 
define $\alpha_0= \inf \supp(\rho_0),\beta_0= \sup \supp(\rho_0)$, $\alpha_1= \inf \supp(\rho_1),\beta_1= \sup \supp(\rho_1)$, and 
set
\[
T^{\rm +}(x)=\left\{
\begin{array}{ll}
T(x) & \mbox{ when }x\in \supp(\rho_0)\\
\beta_1 + x - \beta_0 & \mbox{ when }x\geq \beta_0\\
\alpha_1 + x - \alpha_0 & \mbox{ when }x \leq \alpha_0.
\end{array}\right.
\]
Let $\gamma_0,\gamma_1 \in \mathbb{R}$ be the uniquely defined points such that $\frac{x_0-\alpha_0}{\f^\star(\alpha_0)}=\frac{x_0-\gamma_0}{T^+(\gamma_0)-\gamma_0}$ and $\frac{x_0-\beta_0}{\f^\star(\beta_0)}=\frac{x_0-\gamma_1}{T^+(\gamma_1)-\gamma_1}$. The point $\gamma_1$ is the point that, when transported by $T-\Id$, crosses the toll at the same time $\beta_0$ crosses the toll when being transported by $\f^\star$. Note that we have $\gamma_0\leq \alpha_0$ and $\gamma_1\geq \beta_0$. We also prolong $\f^\star$ on the whole $\mathbb{R}$ as 
\[
\f^{\star \rm +}(x)=\left\{
\begin{array}{ll} 
T(x)-x & \mbox{ when }x\leq \gamma_0 \mbox{ or } x\geq \beta_0\\
\f^\star(\alpha_0)\frac{x_0-x}{x_0-\alpha_0} & \mbox{ when }\gamma_0\leq x \leq \alpha_0 \\
\f^\star(x) & \mbox{ when }x \in \supp(\rho_0)\\
\f^\star(\beta_0)\frac{x_0-x}{x_0-\beta_0} & \mbox{ when }\beta_0\leq x \leq \gamma_1
\end{array}\right.
\]

For notational simplicity, in the sequel, we suppress the labeling on $T^+$,$\f^{\star +}$ and use $T$, $\f^\star$ instead for the prolonged versions as well.
To build $\f^\star$, we first establish that on the points where $T-\Id$ doesn't satisfy the constraint, $\f^\star$ actually saturates the constraint.
\blue
As an immediate consequence of Proposition \ref{prop2}, we have the following lemma:
\black
\begin{lemma}\label{lemma1}
For all $x\in \supp(\rho_0)$ such that $C_x(T-\Id)> h$ we have $C_x(\f^\star)= h$.
\end{lemma}

We next characterize a leading  segment of the distribution corresponding to points with velocity faster than that of the optimal unconstrained transport.
It is essential that the leading edge ``speeds up'' to allow the trailing portion to pass through and meet the time constraint. Specifically, we show that $\f^\star$ is greater than $T-\Id$ at the points to the right of points where $T-\Id$ doesn't satisfy the constraint.

\begin{lemma}\label{lemma2}
For $x_1=\sup \{x\in \supp(\rho_0)\ | \ C_x(T-\Id)> h\} $ and $y_1=\sup \{x \in \mathbb{R}\ | \ C_x^{\rm alt}({\f^\star})= h\} $ we have that $\forall x \in (x_1,y_1)$, $\f^\star(x)\geq T(x)-x$.
\end{lemma}

\begin{proof}
First note that $y_1\geq x_1$ by Lemma \ref{lemma1}. Suppose that 
\[
\{\f^\star(x)<T(x)-x\}\cap(x_1,y_1)\neq \emptyset
\]
and let
$
a=\sup \{x \in (x_1,y_1) \ | \ \f^\star(x)<T(x)-x\}
$. We consider separately the two cases $\rho_0(a)=0$ and  $\rho_0(a)>0$ below:\\[.03in]
i) If $\rho_0(a)=0$ then $\forall x\geq a,\ \rho_0(x)=0$, so
\[
\f^\star(x) = \frac{T(y_1)-y_1}{y_1-x_0}(x-x_0),
\]
as $C_x^{\rm alt}({\f^\star})= h$ for all x$\in [a,y_1]$. Furthermore, $T(a)-a=\beta_1+a-\beta_0-a=T(y_1)-y_1$ and $T(a)-a=\f^\star(a)$ so necessarily $a=y_1$ and $\supp(\rho_0)=[\alpha_0,y_1]$. Then $\exists z\in (x_1,y_1)$ such that, $\rho_0(z)>0$, $\f^\star(z)<T(z)-z$ and $\partial_x\f^{\star}(z)>T^\prime(z)-1$.\\[.03in]
ii) If $\rho_0(a)>0$, then by convexity of $\supp(\rho_0)$ we also have existence of that $z\in (x_1,y_1)$ with the same properties. In both cases we have
\[
\frac{\f^\star(z)\rho_0(z)}{1+\frac{x_0-z}{\f^\star(z)}\partial_x\f^{\star}(z)}< \frac{(T(z)-z)\rho_0(z)}{1+\frac{x_0-z}{T(z)-z}(T^\prime(z)-1)} \leq h
\]
which contradicts the definition of $y_1$.
\end{proof}

The following lemma states that if $\f^\star$ saturates the constraint on a maximal interval (i.e., such that, the points just outside do not saturate the constraint), then either $\f^\star=T-\Id$ throughout, or it is strictly greater than $T-\Id$ on a portion of the interval and strictly less than $T-\Id$ on another portion of the interval. This property is inherited by the convexity of the cost.

\begin{lemma}\label{lemma3} For $[a,b] \subset \{x \in \mathbb{R}\ | \ C_x^{\rm alt}({\f^\star})= h\}$ with $[a,b]$ of maximal size,  
$\exists x\in [a,b]$ such that $\f^\star(x)>T(x)-x$ if and only if
$\exists y\in [a,b]$ such that $\f^\star(y)<T(y)-y$.
\end{lemma}

\begin{proof} Suppose that $\forall x\in [a,b]$ we have $\f^\star(x)\geq T(x)-x$ and we don't have equality on the whole interval.  Define $$\Psi_a(\epsilon) = \int_a^b \big((x_0-x)(\f^\star(x)+\epsilon)+\frac{(T(x)-x_0)^2}{1-\frac{x_0-x}{\f^\star(x)+\epsilon}}\big)\rho_0(x)dx.$$
Then we have $\partial_x\Psi_a(0) = \int_a^b (x_0-x)(1-\frac{(T(x)-x_0)^2}{(\f^\star(x)-(x_0-x))^2}\big)\rho_0(x)dx>0.$ Let be $c<a$ such that $\partial_x\Psi_c(0)>0$ and $\exists\delta>0$ with $\frac{\f^\star(c)^2\rho_0(c)-h\f^\star(c)}{h(x_0-c)}-\partial_x\f^{\star}(c) = -\delta$. Then there exist $d\in(c,a)$ such that $\partial_x\Psi_d(0)>0$ and $\exists\delta>0$ with $\frac{\f^\star(d)^2\rho_0(d)-h\f^\star(d)}{h(x_0-d)}-\partial_x\f^{\star}(d) = -\delta/2$. Let us define $k_\epsilon$ as the function solving the ODE
$$\left\{\begin{array}{ll}
  \partial_xk_\epsilon(x) = -\partial_x\f^{\star}(x) +  \frac{(\f^\star(x)+k_\epsilon)^2\rho_0(x)-h(\f^\star(x)+k_\epsilon(x))}{h(x_0-x)} \quad \quad \text{ for } x\leq d,\\
    k_\epsilon(d)=-\epsilon.
\end{array}\right.$$ 
Then for $\epsilon>0$ small enough we have $\partial_xk_\epsilon(x)<-\delta/4$, $\forall x\in (c,d)$. Therefore for $\epsilon>0$ small enough $\exists y \in (c,d)$ such that $k_\epsilon(y)=0$. Define
\begin{equation*} 
\label{bar_a}
\f_\epsilon(x) = \left\{
\begin{array}{ll}
 \f^\star(x) & \text{if } x\notin (y,b), \medskip\\
    \f^\star(x)-\epsilon & \text{if } x\in (d,b), \medskip\\
    \f^\star(x)+k_\epsilon(x) & \text{if } x\in (y,d]. \medskip
\end{array}
\right.
\end{equation*}  Then for $\epsilon>0$ small enough, $\f_\epsilon$ verifies the constraint and $J(\f_\epsilon)< J(\f^\star)$. Using the same method we can prove that having $\f^\star(x)\leq T(x)-x$ for all $ x\in [a,b]$ is impossible.
\end{proof}

We are now in a position to build explicitly $\f^\star$ using the lemmas. The process of building $\f^*$ consists of determining its value successively on intervals $[z_{y_i},y_i]$ and $[y_{i+1},z_{y_i}]$,
with
\[
\ldots> y_i> z_{y_i} > y_{i+1}>z_{y_{i+1}}>\ldots
\]
such that $\f^\star(x)\neq T(x)-x$ for $x \in [z_{y_i},y_i]$ a.e., while $\f^\star(x) = T(x) -x$ on the complement where $x \notin \bigcup \limits_i [z_{y_i},y_i]$. By Proposition \ref{prop2} we know that $C_x^{\rm alt}({\f^\star})=h$ on intervals $[z_{y_i},y_i]$, a fact that will help us determine $\f^\star$ and the succession of points that define these intervals.

\begin{figure}[!htb]
\centering
\tikzset{every picture/.style={line width=0.75pt}}       
\scalebox{0.9}{
\begin{tikzpicture}[x=0.75pt,y=0.75pt,yscale=-1,xscale=1]
\draw [line width=1.5]    (40.33,118.33) -- (456,118);
\draw [shift={(459,118)}, rotate = 179.95] [color={rgb, 255:red, 0; green, 0; blue, 0 }  ][line width=1.5]    (14.21,-4.28) .. controls (9.04,-1.82) and (4.3,-0.39) .. (0,0) .. controls (4.3,0.39) and (9.04,1.82) .. (14.21,4.28);
\draw [line width=1.5]    (89.33,118) .. controls (95.33,59) and (108.33,21) .. (123.33,68) .. controls (138.33,115) and (123.67,111.67) .. (189.67,111.67) .. controls (255.67,111.67) and (242.33,107) .. (266.33,39) .. controls (290.33,-29) and (296.33,118.67) .. (340.33,117.67);
\draw [line width=1.5]    (128.33,110) -- (128.33,117.67);
\draw [line width=1.5]    (161,111) -- (161,118);
\draw [line width=1.5]    (218.33,110.33) -- (218.33,118.33);
\draw [line width=1.5]    (303,111.33) -- (303,118.33);
\draw [line width=1.5]    (47.67,111.67) -- (47.67,117.67);
\draw [line width=1.5]    (89.67,109.67) -- (89.33,118);
\draw [line width=1.5]    (340.33,111.67) -- (340.33,117.67);
\draw [line width=1.5]    (383.67,112.33) -- (383.67,117);
\draw [line width=1.5]    (438.33,113) -- (438.33,119);

\draw (459,118) node [anchor=north west][inner sep=0.75pt]  [font=\Large]  {$x$};
\draw (428.33,130) node [anchor=north west][inner sep=0.75pt] [font=\Large]    {$x_{0}$};
\draw (330.67,130) node [anchor=north west][inner sep=0.75pt] [font=\Large]    {$\beta_{0}$};
\draw (40,130) node [anchor=north west][inner sep=0.75pt] [font=\Large]   {$z_{y_2}$};
\draw (77.67,130) node [anchor=north west][inner sep=0.75pt] [font=\Large]   {$\alpha_{0}$};
\draw (373.33,130) node [anchor=north west][inner sep=0.75pt] [font=\Large]    {$y_{1}$};
\draw (340.67,19.33) node [anchor=north west][inner sep=0.75pt] [font=\LARGE] {$\rho _{0}$};
\draw (117.33,130) node [anchor=north west][inner sep=0.75pt] [font=\Large]   {$x_{2}$};
\draw (151.33,130) node [anchor=north west][inner sep=0.75pt]  [font=\Large]  {$y_{2}$};
\draw (205.33,130) node [anchor=north west][inner sep=0.75pt] [font=\Large]   {$z_{y_{1}}$};
\draw (290.33,130) node [anchor=north west][inner sep=0.75pt] [font=\Large]   {$x_{1}$};
\end{tikzpicture}
}
\caption{Density $\rho_{0}(x)$ vs.\ $x$}
\label{fig:Rho0} 

\medskip
\centering
\tikzset{every picture/.style={line width=0.75pt}} 
\scalebox{0.9}{
\begin{tikzpicture}[x=0.75pt,y=0.75pt,yscale=-1,xscale=1]

\draw [line width=1.5]    (35,229.67) -- (458.33,230.66);
\draw [shift={(461.33,230.67)}, rotate = 180.13] [color={rgb, 255:red, 0; green, 0; blue, 0 }  ][line width=1.5]    (14.21,-4.28) .. controls (9.04,-1.82) and (4.3,-0.39) .. (0,0) .. controls (4.3,0.39) and (9.04,1.82) .. (14.21,4.28);
\draw [line width=1.5]    (49,223.67) -- (49,229.67);
\draw [line width=1.5]    (437,224.33) -- (437,229.67);
\draw    (394.5,202.83) -- (394.19,227.83);
\draw [shift={(394.17,229.83)}, rotate = 270.71] [color={rgb, 255:red, 0; green, 0; blue, 0 }  ][line width=0.75]    (10.93,-3.29) .. controls (6.95,-1.4) and (3.31,-0.3) .. (0,0) .. controls (3.31,0.3) and (6.95,1.4) .. (10.93,3.29);
\draw    (394.17,229.83) -- (394.48,204.83);
\draw [shift={(394.5,202.83)}, rotate = 90.71] [color={rgb, 255:red, 0; green, 0; blue, 0 }  ][line width=0.75]    (10.93,-3.29) .. controls (6.95,-1.4) and (3.31,-0.3) .. (0,0) .. controls (3.31,0.3) and (6.95,1.4) .. (10.93,3.29);

\draw  [dash pattern={on 4.5pt off 4.5pt}]  (85.67,202.33) -- (418.67,201.83);
\draw [line width=1.5]    (85.67,202.33) -- (85.67,229.67);
\draw [line width=1.5]    (85.67,202.33) -- (201.67,202.33);
\draw [line width=1.5]    (201.67,202.33) -- (202.17,222.33);
\draw [line width=1.5]    (272.33,202.33) -- (273.17,222.33);
\draw [line width=1.5]    (272.33,202.33) -- (359,202.33);
\draw [line width=1.5]    (359,202.33) -- (359,230.33);
\draw [line width=1.5]    (202.17,222.33) .. controls (243.17,229.83) and (234.67,221.83) .. (273.17,222.33);

\draw (461.33,230.67) node [anchor=north west][inner sep=0.75pt]  [font=\Large]  {$t$};
\draw (41.33,240) node [anchor=north west][inner sep=0.75pt] [font=\Large]   {$0$};
\draw (429.33,240) node [anchor=north west][inner sep=0.75pt] [font=\Large]   {$1$};
\draw (399.33,206) node [anchor=north west][inner sep=0.75pt] [font=\Large]   {$h$};
\draw (57.83,240) node [anchor=north west][inner sep=0.75pt]   {$toll( \beta _{0})$};
\draw (171.67,240) node [anchor=north west][inner sep=0.75pt]    {$toll( z_{y}{}_{_{1}})$};
\draw (248.17,240) node [anchor=north west][inner sep=0.75pt]     {$toll(y_{2})$};
\draw (331.67,240) node [anchor=north west][inner sep=0.75pt]   {$toll( \alpha _{0})$};
\end{tikzpicture}
}
\caption{Flux $\rho_{t}(x_0)v_{t}(x_0)$ at crossing.}
\label{fig:momentum}
    
\medskip
\centering
\tikzset{every picture/.style={line width=0.75pt}} 
\scalebox{0.8}{
\begin{tikzpicture}[x=0.75pt,y=0.75pt,yscale=-1,xscale=1]
\draw [line width=1.5]    (79,353.47) -- (596,353.47);
\draw [shift={(599,353.47)}, rotate = 180] [color={rgb, 255:red, 0; green, 0; blue, 0 }  ][line width=1.5]    (14.21,-4.28) .. controls (9.04,-1.82) and (4.3,-0.39) .. (0,0) .. controls (4.3,0.39) and (9.04,1.82) .. (14.21,4.28);
\draw [line width=1.5]    (322.93,110.2) -- (325.97,384.8);
\draw [shift={(326,387.8)}, rotate = 269.37] [color={rgb, 255:red, 0; green, 0; blue, 0 }  ][line width=1.5]    (14.21,-4.28) .. controls (9.04,-1.82) and (4.3,-0.39) .. (0,0) .. controls (4.3,0.39) and (9.04,1.82) .. (14.21,4.28);
\draw [line width=1.5]    (324,152.47) -- (332,152.6);
\draw [line width=1.5]    (79,133.47) -- (596,133.47);
\draw [shift={(599,133.47)}, rotate = 180] [color={rgb, 255:red, 0; green, 0; blue, 0 }  ][line width=1.5]    (14.21,-4.28) .. controls (9.04,-1.82) and (4.3,-0.39) .. (0,0) .. controls (4.3,0.39) and (9.04,1.82) .. (14.21,4.28);
\draw  [color={rgb, 255:red, 74; green, 144; blue, 226 }  ,draw opacity=1 ][fill={rgb, 255:red, 74; green, 144; blue, 226 }  ,fill opacity=1 ][line width=1.5]  (341.14,128.39) -- (342.09,192.66) -- (326.13,218.6) -- (325.19,154.33) -- cycle;
\draw  [color={rgb, 255:red, 208; green, 2; blue, 27 }  ,draw opacity=1 ][fill={rgb, 255:red, 208; green, 2; blue, 27 }  ,fill opacity=1 ][line width=1.5]  (342.9,236.55) -- (343.63,297.91) -- (327.28,324.35) -- (326.54,262.99) -- cycle;
\draw  [fill={rgb, 255:red, 74; green, 144; blue, 226 }  ,fill opacity=1 ][line width=1.5]  (245.93,61.53) .. controls (254.23,61.2) and (269.43,133.33) .. (276,133.8) .. controls (282.57,134.27) and (153.84,134.04) .. (152.77,133.67) .. controls (151.69,133.29) and (152.51,128.71) .. (152.51,128.38) .. controls (152.51,128.04) and (207.63,128.89) .. (218.39,124.39) .. controls (229.15,119.89) and (230.36,106.58) .. (231.64,99.14) .. controls (232.93,91.7) and (237.64,61.87) .. (245.93,61.53) -- cycle;
\draw  [fill={rgb, 255:red, 208; green, 2; blue, 27 }  ,fill opacity=1 ][line width=1.5]  (143.77,128) .. controls (150.87,128.44) and (151.63,128.26) .. (152.51,128.38) .. controls (153.39,128.5) and (152.87,134.11) .. (152.77,133.67) .. controls (152.67,133.22) and (99.27,133) .. (97.77,133.5) .. controls (96.27,134) and (106.1,82.86) .. (118.76,82.67) .. controls (131.43,82.48) and (136.67,127.56) .. (143.77,128) -- cycle;
\draw  [fill={rgb, 255:red, 74; green, 144; blue, 226 }  ,fill opacity=1 ][line width=1.5]  (508,332.67) .. controls (515.5,323.17) and (516.93,297.32) .. (532.9,297.05) .. controls (548.88,296.77) and (564.17,352.5) .. (563.8,353.13) .. controls (563.43,353.77) and (440.17,353) .. (439.8,353.13) .. controls (439.43,353.27) and (439.71,347.14) .. (439.67,347) .. controls (439.62,346.86) and (489.5,345.67) .. (494,343.67) .. controls (498.5,341.67) and (500.5,342.17) .. (508,332.67) -- cycle;
\draw  [fill={rgb, 255:red, 208; green, 2; blue, 27 }  ,fill opacity=1 ][line width=1.5]  (422.58,336.61) .. controls (430.08,342.61) and (440.05,347.81) .. (439.67,347) .. controls (439.29,346.19) and (439.8,354.13) .. (439.8,353.13) .. controls (439.8,352.13) and (367.9,353.09) .. (371.85,352.94) .. controls (375.8,352.8) and (378.82,311.46) .. (396.42,310.66) .. controls (414.02,309.86) and (415.08,330.61) .. (422.58,336.61) -- cycle;
\draw [line width=1.5]    (325,330.47) -- (331.42,320.82);
\draw [line width=1.5]    (277.73,128.2) -- (277.8,132.6);
\draw [line width=1.5]    (324,152.47) -- (328.22,146.42);
\draw [line width=1.5]    (211.73,128.2) -- (211.8,133.6);
\draw [line width=1.5]    (277.73,128.2) -- (277.8,133.6);
\draw [line width=1.5]    (96.73,127.2) -- (96.8,132.6);
\draw [line width=1.5]    (373.6,348.6) -- (373.67,354);
\draw [line width=1.5]    (492.13,347.4) -- (492.2,352.8);
\draw [line width=1.5]    (563.73,347.73) -- (563.8,353.13);
\draw [line width=1.5]    (276,133.8) -- (324,152.47) ;
\draw [shift={(307.27,145.96)}, rotate = 201.25] [color={rgb, 255:red, 0; green, 0; blue, 0 }  ][line width=1.5]    (14.21,-4.28) .. controls (9.04,-1.82) and (4.3,-0.39) .. (0,0) .. controls (4.3,0.39) and (9.04,1.82) .. (14.21,4.28);
\draw [line width=1.5]    (324,152.47) -- (563.8,353.13);
\draw [shift={(449.88,257.81)}, rotate = 219.92] [color={rgb, 255:red, 0; green, 0; blue, 0 }  ][line width=1.5]    (14.21,-4.28) .. controls (9.04,-1.82) and (4.3,-0.39) .. (0,0) .. controls (4.3,0.39) and (9.04,1.82) .. (14.21,4.28);
\draw  [color={rgb, 255:red, 74; green, 144; blue, 226 }  ,draw opacity=1 ][fill={rgb, 255:red, 74; green, 144; blue, 226 }  ,fill opacity=1 ] (331,210.59) .. controls (332.33,214.03) and (331.67,212.93) .. (330.33,229.26) .. controls (329,245.59) and (332,251.59) .. (330.67,253.93) .. controls (329.33,256.26) and (325.6,263.15) .. (326,260.59) .. controls (326.4,258.04) and (324.93,222.2) .. (326.13,218.6) .. controls (327.33,215) and (329.67,207.16) .. (331,210.59) -- cycle;
\draw [line width=1.5]    (324.47,220.27) -- (329,213.59);
\draw [line width=1.5]    (324.67,263.8) -- (329,256.59);
\draw [line width=1.5]    (211.8,133.6) -- (324.47,220.27);
\draw [shift={(274.32,181.69)}, rotate = 217.57] [color={rgb, 255:red, 0; green, 0; blue, 0 }  ][line width=1.5]    (14.21,-4.28) .. controls (9.04,-1.82) and (4.3,-0.39) .. (0,0) .. controls (4.3,0.39) and (9.04,1.82) .. (14.21,4.28);
\draw [line width=1.5]    (324.47,220.27) -- (492.2,352.8) ;
\draw [shift={(414.45,291.37)}, rotate = 218.31] [color={rgb, 255:red, 0; green, 0; blue, 0 }  ][line width=1.5]    (14.21,-4.28) .. controls (9.04,-1.82) and (4.3,-0.39) .. (0,0) .. controls (4.3,0.39) and (9.04,1.82) .. (14.21,4.28);
\draw [line width=1.5]    (152.77,133.67) -- (324.67,263.8) ;
\draw [shift={(244.94,203.44)}, rotate = 217.13] [color={rgb, 255:red, 0; green, 0; blue, 0 }  ][line width=1.5]    (14.21,-4.28) .. controls (9.04,-1.82) and (4.3,-0.39) .. (0,0) .. controls (4.3,0.39) and (9.04,1.82) .. (14.21,4.28);
\draw [line width=1.5]    (324.67,263.8) -- (441.62,354.19) ;
\draw [shift={(389.31,313.77)}, rotate = 217.7] [color={rgb, 255:red, 0; green, 0; blue, 0 }  ][line width=1.5]    (14.21,-4.28) .. controls (9.04,-1.82) and (4.3,-0.39) .. (0,0) .. controls (4.3,0.39) and (9.04,1.82) .. (14.21,4.28);
\draw [line width=1.5]    (97.77,133.5) -- (325,330.47) ;
\draw [shift={(217.28,237.09)}, rotate = 220.92] [color={rgb, 255:red, 0; green, 0; blue, 0 }  ][line width=1.5]    (14.21,-4.28) .. controls (9.04,-1.82) and (4.3,-0.39) .. (0,0) .. controls (4.3,0.39) and (9.04,1.82) .. (14.21,4.28);
\draw [line width=1.5]    (325,330.47) -- (373.67,354) ;
\draw [shift={(356.36,345.63)}, rotate = 205.81] [color={rgb, 255:red, 0; green, 0; blue, 0 }  ][line width=1.5]    (14.21,-4.28) .. controls (9.04,-1.82) and (4.3,-0.39) .. (0,0) .. controls (4.3,0.39) and (9.04,1.82) .. (14.21,4.28);

\draw (269.63,58.56) node [anchor=north west][inner sep=0.75pt]  [font=\LARGE]  {$\rho_{0}$};
\draw (556.93,270.59) node [anchor=north west][inner sep=0.75pt]  [font=\LARGE]  {$\rho_{1}$};
\draw (596,104.47) node [anchor=north west][inner sep=0.75pt]  [font=\LARGE]  {$x$};
\draw (596.5,323.97) node [anchor=north west][inner sep=0.75pt]  [font=\LARGE]  {$x$};
\draw (338.36,375) node [anchor=north west][inner sep=0.75pt]  [font=\LARGE]  {$t$};
\draw (324.5,98.67) node [anchor=north west][inner sep=0.75pt]  [font=\LARGE]  {$x_{0}$};
\draw (140.79,140) node [anchor=north west][inner sep=0.75pt]    [font=\Large] {$y_{2}$};
\draw (195.88,140) node [anchor=north west][inner sep=0.75pt]    [font=\Large] {$z_{y_1}$};
\draw (274.5,217) node [anchor=north west][inner sep=0.75pt]    {$toll(z_{y_{1}})$};
\draw (278.77,263) node [anchor=north west][inner sep=0.75pt]    {$toll(y_{2})$};
\draw (407.75,364.89) node [anchor=north west][inner sep=0.75pt]    [font=\Large] {$\mathrm{Y}_{1}(y_2)$};
\draw (474.2,364.38) node [anchor=north west][inner sep=0.75pt]    [font=\Large] {$\mathrm{Y}_1(z_{y_{2}})$};
\end{tikzpicture}}
\caption{Illustration of the flow through the toll.
The middle segment $[y_2,z_{y_1}]$ transports through the toll unimpeded by the constraint towards the final destination, via the optimal transport map $T$, designed for unconstrained transport; each point in this interval maintains the same velocity before and after the toll. In contrast, the segments to the left and right, $[z_{y_2},y_2]$
and $[z_{y_1},y_1]$, respectively, are adjusted accordingly so as to saturate the constraint. The exact position of their respective end points (that may even be outside the support of $\rho_0$, as a matter of computational simplicity, in which case they correspond to zero density) are computed via the solution of an optimization problem and depend on the terminal distribution $\rho_1$ as well.
}
\label{fig:tollinterval}
\end{figure}

We explain the process in Figures \ref{fig:Rho0}-\ref{fig:tollinterval}
with an example. This example presents a situation  
where the behavior of the corresponding optimal solution $\f^\star$ is characterized by two distinct intervals $[z_{y_i},y_i]$ $i=1,2$, where the constraint saturates.
Thus, for this example, we identify three intervals of interest, $[z_{y_2},y_2]$, $[y_2,z_{y_1}]$, and $[z_{y_1},y_1]$. In the first and the last, the constraint saturates, whereas in the middle interval it does not.
We proceed by working our way from right to left, always assuming that $\supp(\rho_0)$ is to the left of the toll, as in the figures.

In general, the process begins by first computing the optimal transport map $T$, without involving the constraint. Then, we identify $x_1$ as the rightmost point where the throughput hits the limit set at $x_0$. Naturally, if the optimal transport map satisfies the throughput constraint, then it is the optimal map and specifies $\f^*$ throughout. Assuming that $x_1$ is finite, then a search to the right of $x_1$, that we explain later on, identifies $y_1$ as the rightmost point where $\f$ needs to 
be adjusted so as to abide by the throughput
constraint while minimizing the transportation cost. In the example depicted in Fig.\ \ref{fig:Rho0}, $y_1$ is shown located to the right of $\beta_0$ 
($=$ the supremum of the support of $\rho_0$), though this is not always the case, and depends on the terminal distribution $\rho_1$ via the optimization problem that specifies $y_1$.
We choose to explain this case, where $y_1$ is to the right of $\beta_0$ so as to highlight that this is indeed possible.

Continuing on with our specific example,
for the interval $[z_{y_1},y_1]$, we have $\f^\star=\f_{y_1}$,  with $\f_y$ defined in equation \eqref{eq:vy} explained below, which ensures that $C_x^{\rm alt}({\f})=h$.  Then, on $[y_2,z_{y_1}]$ we have once again that the velocity is specified by the ``unconstrained'' optimal map $T$, i.e, that $\f^\star=T-\Id$, and so $C_x^{\rm alt}({\f})=C_x^{\rm alt}({T-\Id})$.
\black
Finally on $[z_{y_2},y_2]$, we have $\f^\star=\f_{y_2}$ as $C_x^{\rm alt}({\f})=h$. Note that in this specific example where $y_1\geq \beta_0$ and $z_{y_2}\leq \alpha_0$, we have \black $\forall x \in [z_{y_2},\alpha_0]$, $toll(x)=toll(\alpha_0)$ and $\forall x \in [\beta_0,y_1]$, $toll(x)=toll(\beta_0)$. 

 We now detail how to build explicitly $\f^\star$ in the general case.  As noted,
if $T-\Id$ verifies the constraint throughout, which can now be explicitly stated as in \eqref{eq:c5}, then $\f^*=T-\Id$ is the optimal solution. Otherwise define $x_1=\sup \{x\in \supp(\rho_0)\ | \ C_x(T-\Id)> h\}$, and thereby we determine $y_1\in[x_1,x_0]$
(cf.\ Lemma \ref{lemma1})
such that
\begin{equation}\label{eq:y1}
y_1=\sup \{x \in \mathbb{R}\ | \ C_x^{\rm alt}({\f^\star})= h\}.
\end{equation} 

For any $y\in \mathbb{R}$ with $ x_1\leq y < x_0$, define the velocity $v_y(x)$ as the solution of the differential equation
\begin{equation} \label{eq:vy}
\left\{\begin{array}{ll}
  \partial_x\f_y(x)=\frac{\f_y(x)^2\rho_0(x)-h\f_y(x)}{h(x_0-x)} \quad \quad \text{ for } x\leq y.\\
    \f_y(y)=T(y)-y
\end{array}\right.
\end{equation}
Note that this equation is solved backwards, starting from a terminal condition at $y$.
This value for the velocity ensures that the transport will saturate the constraint to left of $y$ (i.e., $C_x^{\rm alt}({\f_y})= h$ will hold for $x\leq y$). The functional form of $v_y(x)$ will be used next to identify the first interval $[z_{y_1},y_1]$, where the velocity will depart form that of the unconstrained transport $T$, via solving a suitable optimization problem to determine $y_1$.
Since we know that the equality $C_x^{\rm alt}({\f^\star})= h$ will be true on a certain interval 
$[z_{y_1},y_1]$, on that interval we will have $\f^\star=v_{y_1}$.

Let
$w_y^{x_1}=\inf \{x\leq x_1 \mid \forall s\in (x,x_1), \f_{y}(s)\geq T(s)-s\}$ 
(well defined by Lemma \ref{lemma2})
and $z_y^{x_1}= \inf \{x\leq w_y^{x_1}\ | \ \forall s\in (x,w_y^{x_1}), \ \f_{y}(s)< T(s)-s\}$.
Then we have that
$\f^\star(x)=\f_{y_1}(x)$, $\forall x \in (z_{y_1}^{x_1},x_1)$ by Lemma \ref{lemma3} and Proposition \ref{prop2}. 

We now determine $y_1$ by solving a suitable optimization problem.
For $x\leq y < x_0$, define 
\begin{align*}
J_x(y) &= \int_0^1 \int_\mathbb{R}
\Big((T(s)-s)
\mathds{1}\{s\notin (z_y^{x},y)\}\\
&\hspace*{15pt} +((x_0-s)\f_{y}(s)+\frac{(T(s)-x_0)^2}{1-\frac{x_0-s}{\f_{y}(s)}})\mathds{1}\{s\in (z_y^{x},y)\}\Big)\rho_0(s)dsdt.
\end{align*} 
We have $J_x(y)=J(\f_y^{+})$ for the function $\f_y^{+}$ such that $\f_y^{+}=\f_y$ on $[z_y^x,y]$ and $\f_y^{+}=T-\Id$ on $\mathbb{R}\setminus  [z_y^x,y]$. Then the first step of the building process of $\f^\star$ is to find $y_1$ solution of 
\[
    y_1= \argmin \limits_{y\geq x_1} J_{x_1}(y).
\]
Such a $y_1$ is well defined as $J_{x_1}$ is continuous on $[x_1,x_0]$. 
Once $y_1$ has been determined, we define $x_2=\sup \{x<z_{y_1} \ | \ C_x(T-\Id)> h\}$. If $x_2$ is not defined then
\[
\f^\star(x) = \left\{
\begin{array}{ll}
    \f_{y_1}(x) & \text{if } x\in (z_1,y_1), \medskip\\
    T(x)-x & \text{if } x\notin (z_1,y_1), \medskip
\end{array}
\right.
\]
otherwise we start again the same process to determine $y_2$ as
\begin{align*}
    y_2= \argmin \limits_{y>x_2} J_{x_2}(y).
\end{align*}
If $y_2<z_{y_1}$, it suggests that there is an interval $[y_2,z_{y_1}]$ where the transport follows the unconstrained map $T$, and we continue in the same way.

However, it is possible that the condition $y_i\leq z_{y_{i-1}}$ fails at some point, for some $i\geq 2$. In that case, intervals where the velocity departs from being $T(x)-x$, will merge. 
For instance, if we obtain $y_i>z_{y_{i-1}}$ then as $(y,y^\prime)\mapsto J_{x_{i-1}}(y) + J_{x_{i}}(y^\prime)$ is convex on $\{(y,y^\prime)\ | \ y^\prime\leq z_y\}$, it means that $C_x^{\rm alt}({\f^\star})= h$, $\forall x \in (x_i,x_{i-1})$ and therefore we have to start the optimization again and determine $y_{i-1}$ as 
\begin{align*}
    y_{i-1}= \argmin \limits_{y>x_{i-1}} J_{x_i}(y).
\end{align*}

If we obtain a value $y_{i-1}>z_{y_{i-2}}$, we reset $x_{i-1}$ as being equal to $x_i$ and, once again, we have to redetermine 
\[
y_{i-2}= \argmin \limits_{y>x_{i-2}} J_{x_{i-1}}(y).
\]
Otherwise, i.e., if we obtain a value
$y_{i-1}\leq z_{y_{i-2}}$,
we reset $x_i$ as  $x_i=\sup \{x<z_{y_{i-1}} \ | \ C_x(T-\Id)> h\}$ for this updated value $y_{i-1}$. 
Once again, if $x_i$ is well defined we continue the process by finding \[
y_i= \argmin \limits_{y>x_i} J_{x_i}(y).
\]
We continue this iterative process until $\f^\star$ is defined on all of the support of $\rho_0$.
We finally remark that \black
$$
E:=\bigcup \limits_{i=1}^n (z_{y_i},y_i)=\{x \in \mathbb{R} \ | \ \exists \delta>0, \forall y\in(x-\delta,x)\cup (x,x+\delta), T(x)-x\neq \f^\star(x)\}.
$$
\blue
Note that we have that $n\in \mathbb{N}\cup \{+\infty\}$, so the process doesn't necessarily terminate. If one absolutely wants the process to terminate, they have to be careful to the oscillations of $C_x(T-\Id)$ around the value $h$. Indeed, if the process doesn't terminate, it implies that $x_i=\sup \{x<z_{y_{i-1}} \ | \ C_x(T-\Id)> h\}$ always exists $\forall i$, so the function $x\mapsto C_x(T-\Id)$ oscillates indefinitely around $h$ as $x$ is moving backward. Supposing that the densities $\rho_0$ and $\rho_1$ are Lipschitz, then $T$ has Lipschitz derivative so $x\mapsto C_x(T-\Id)$ is also Lipschitz. This implies that the oscillations around $h$ become smaller and smaller (in size) so it suffices to lower the value of $h$ of any $\epsilon>0$ to avoid the infinite oscillations.
\black

We summarize our conclusions on the shape of $\f^\star$ in the following statement.
\begin{theorem}\label{thm:4.5} The solution $\f^\star$ of Problem \ref{prob:problem3} satisfies:
\begin{itemize}
    \item[a)] For $x\in (z_{y_i},y_i)$, we have $\f^\star(x)=\f_{y_i}(x)$ (defined in \eqref{eq:vy}) and  $C_x^{\rm alt}({\f^\star})= h$.
    \item[b)] There exists $ w_i\in (z_{y_i},y_i)$ such that $\forall y \in[w_i,y_i)$, $\f^\star(y)\geq T(y)-y$ and $\forall y \in(z_{y_i},w_i]$, $\f^\star(y)\leq T(y)-y$.
    \item[c)] $\forall x \notin E$, $\f^\star(x)= T(x)-x$.
    \item[d)] The building process of $\f^\star$ consists in solving iteratively $y_i= \argmin \limits_{y>x_i} J_{x_i}(y)$ and if $y_{j}>z_{y_{j-1}}$, then starting again by setting $y_{j-1}= \argmin \limits_{y>x_{j-1}} J_{x_{j}}(y)$. \blue  The process may contain an infinite number of steps. \black 
\end{itemize}
\end{theorem}

\section{Numerical example}\label{sec:example} 
We provide an example to highlight the departure  of the optimal transport plan through a toll with a bound on the flux, from the ideal unconstrained transport $T$. The example we have selected is basic, with
uniform probability densities $\rho_0(x) = \mathds{1}\{x \in [0,1]\}$, $\rho_1(x) = \mathds{1}\{x \in [2,3]\}$ , and a toll at $x_0=3/2$ with a bound $h$ on the flux, with  $1<h\leq 2$. The stringent constraint on the flux, that necessitates varying velocities so as to redistribute the mass flow as it traverses the toll, is clearly seen in the succession of distributions $Y_{t\sharp \rho_0}$ displayed in Fig.~\ref{fig:melange}. Evidently, these readily contrast with the unconstrained transport that pushes forward $\rho_0$ with constant speed giving $\rho_t(x)=\rho_0(x-2t)$.

Specifically, with the flux-constraint in place, we obtain that the optimal transport is effected by
\[
Y_t(x)=\left\{\begin{array}{lll}
x +t\f(x) &\mbox{ for }&t\leq \toll(x)=\frac{3/2-x}{\f(x)},\\\\
3/2 +(t-\toll(x))g(x) &\mbox{ for }&t\geq \toll(x).
\end{array}\right.
\]
Then, the constraint \eqref{eq:c5} gives that $\f$ solves the ODE 
\[
\frac{\f(x)}{1+\frac{3/2-x}{\f(x)} \partial_x\f(x)} =h.
\]
It follows that $\f(x) = \frac{h(2x-3)}{2x-3+\alpha}$ for a certain value $\alpha \in \mathbb{R}$ \black. Using the fact that the optimal solution must be symmetric in time $(\f(x)=g(1-x))$ and that $g(x) = \frac{x+0.5}{1-\toll(x)}$, we finally obtain that $\f(x) = \frac{h(2x-3)}{2x-1-h}$.
Snapshots of the flow along the path from $\rho_0$ to $\rho_1$ are depicted in Figure \ref{fig:melange}.

\begin{figure}[htb]
  \includegraphics[width=1.1\textwidth]{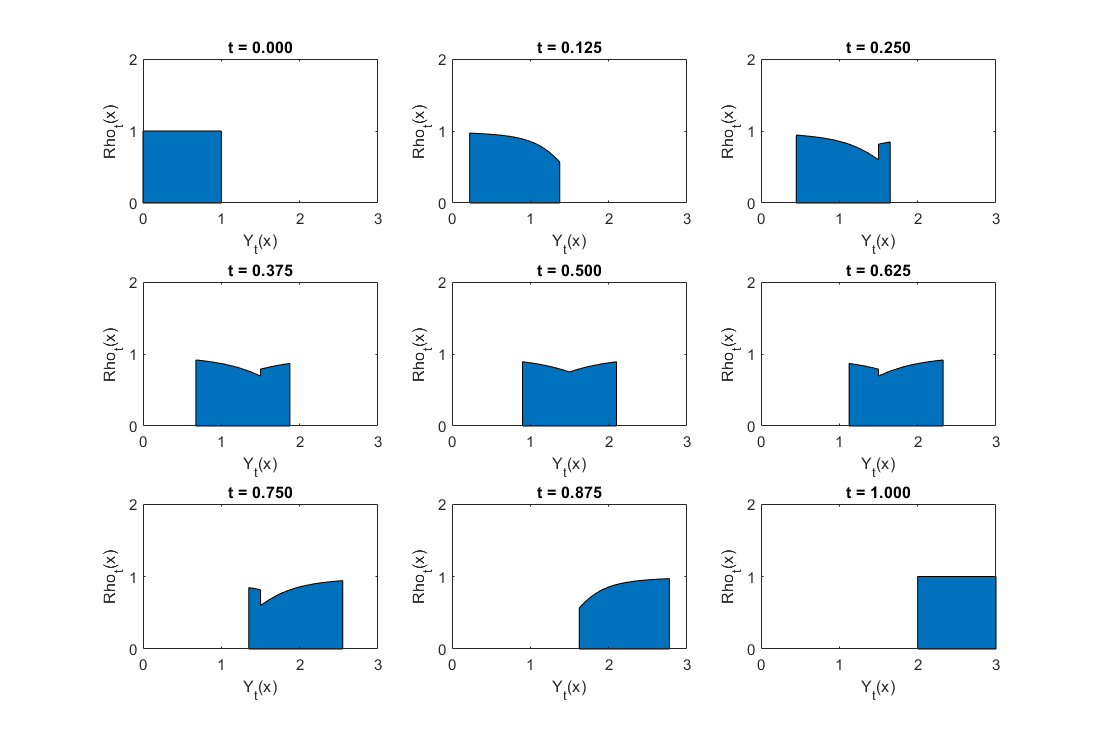}  
\caption{ Example of  transporting a uniform distribution through a constriction (with $h=1.5$) to a similar uniform terminal distribution. While the optimal unconstrained transport will preserve the shape of the marginals at each time $t$, the flux constraint necessitates an optimal velocity that changes with $x$, stretching the leading edge of the distribution as it approaches the toll. Note that the snapshots of the transported distributions $Y_{t\sharp \rho_0}$ ``squeeze'' while crossing the toll, and that the flow is symmetric with time.}
\label{fig:melange}
\end{figure}

\section{Discussion and conclusion}\label{sec:discussion}
We have presented theory for the most basic optimal transport problem in $\mathbb R$, through a constriction where a throughput constraint is imposed. We modeled the formulation after the standard Monge-Kantorovich optimal transport with a quadratic cost. We have shown that an optimal transport exists and is unique under general assumptions. Under some suitable assumption on the densities to be transported to one another, we have shown explicitly how to construct the transport plan. Moreover, we have highlighted natural properties of the transport plan.

More generally, in the case where $\rho_0$ and $\rho_1$ are densities on $\mathbb{R}^d$
and that all the trajectories have to pass through a single point $x_0\in \mathbb{R}^d$, we can readily extend the result presented as follows.
For $\lambda_{\alpha S^{d-1}}$ the Lebesgue measure on the sphere of radius $\alpha$ and center $x_0$, 
define
$$
\nu^0(\alpha)= \int_{\alpha S^{d-1}}\rho_0(x)d\lambda_{\alpha S^{d-1}}(x)
$$
and $\nu^1$ the same way. Then the problem in $\mathbb{R}^d$ is equivalent to solving the problem in dimension 1 between the measure $\nu^0,\nu^1$ defined as $\nu^0(x)=\mathds{1}\{x<0\} \nu^0(-x)$ and $\nu^1(x)=\mathds{1}\{x>0\} \nu^1(x)$.

A significant departure from the current setting arises in the case of multiple tolls, or of a continuum of tolls, where the flux-rate is bounded on a curve, surface, etc.\ The case where a sequence of tolls, possibly even zero-dimensional (points), where mass has to flow through all in succession, is of particular interest in engineering applications. Indeed, in the modern information age, knowledge of obstructions ``down the road'' can undoubtedly be used to optimize transportation cost upstream.
On the other hand, the paradigm of multiple alternative tolls that one can choose to cross, is expected to have a more combinatorial flavor. Lastly, one could generalize the problem presented in this paper to transport of densities in dimension $d$, with a flux constraint on a measurable set with respect to the $p$-dimensional Hausdorff measure $\mathcal{H}^p$ (with $p\leq d$). For instance, an analogous flux constraint on a measurable set $A\subset \mathbb{R}^d$ with $0<\mathcal{H}^p(A)<\infty$ can be cast as: $\forall B\subset A$ measurable with $\mathcal{H}^p(B)>0$ and $t\in (0,1)$
\begin{equation*}
\forall \alpha_1,\alpha_2\in \mathbb{R}, \quad \frac{\mathcal{H}^p(A)}{\mathcal{H}^p(B)}  \int \mathds{1}_{\{\exists \tau \in (t+\alpha_1,t+\alpha_2)\ |\ Y_\tau(x)\in B\}} \rho_0(x)dx \leq h|\alpha_2-\alpha_1|.
\end{equation*}
The proof of existence and uniqueness of a solution should follow using similar arguments. However, to completely characterise the behavior of the solution as in the simpler case treated herein, is expected to be considerably more challenging; one would need a finer description of how the mass distributes while traversing the toll. \color{black}



Transport problems with a throughput restriction are quite natural in a variety of scientific disciplines. 
Of course, transportation through tolls on highways represents perhaps the most rudimentary paradigm in an engineering setting.
Likewise, throughput through servers with a throughput bound is common in queuing systems. A continuum theory as envisioned herein, in higher dimension and with multiple serial tolls, may produce useful practical insights.
Finally, while fluid flow, passing through constrictions or porous media, though not directly abiding by the rigid setting of bounded throughput, could provide an idealized pertinent model in certain situations. 
Evidently, for an accurate model for fluid past constrictions, besides distinguishing between compressible and incompressible, throughput must be dictated by pressure, which in turn may be introduced in a suitable cost functional to be optimized for a further broadening of the general program.



\section*{Acknowledgments}
The research was also supported in part by the National Science Foundation under grant 1807664 and the Air Force Office of Scientific Reserarch under grant FA9550-20-1-0029.

\bibliographystyle{siam}
\bibliography{ref}

\end{document}